%% file: Relay_AoI.tex
\documentclass[10pt,journal,twoside]{IEEEtran}
\IEEEoverridecommandlockouts
\ifCLASSOPTIONcompsoc
\usepackage[nocompress]{cite}
\else
\usepackage{balance}
\usepackage{cite}
\fi
\ifCLASSINFOpdf
\usepackage[pdftex]{graphicx}
\graphicspath{{../Figures/}}
\DeclareGraphicsExtensions{.pdf,.jpeg,.png}
\else
\usepackage{graphicx}

\fi

\usepackage{amsthm}
\usepackage{amsmath,amssymb,amsfonts}
\usepackage{algorithmic}
\usepackage[ruled,vlined,linesnumbered]{algorithm2e}
\usepackage{smartdiagram}
\usepackage{geometry}
\usepackage{epsfig}
\usepackage{color}
\usepackage{lettrine}
\usepackage{float}
\usepackage{lipsum}
\usepackage{bm}
\usepackage{subcaption}
\usepackage{mathtools}
\usepackage{bbm}
\usepackage{etoolbox}
\usepackage{suffix}
\usepackage[T1]{fontenc}
\usepackage[colorlinks]{hyperref}
\usepackage[nameinlink,noabbrev]{cleveref}
\usepackage{textcomp}
\usepackage{xcolor}
\usepackage{multicol}
\usepackage{mdwmath}
\usepackage{mdwtab}
\usepackage{stackengine}
\usepackage{cancel}
\allowdisplaybreaks[3]

\usepackage{dsfont}

\input{Supporting_Preambles/symbols_commands}

\def\argmin{\operatornamewithlimits{arg\,min}}

\newcommand{\bE}{\mathds{E}}

\newcommand{\ol}{\overline}
\renewcommand{\wt}{\widetilde}

\definecolor{lime}{HTML}{A6CE39}
\DeclareRobustCommand{\orcidicon}{
\begin{tikzpicture}
\draw[lime, fill=lime] (0,0)
circle[radius=0.16]
node[white]{{\fontfamily{qag}\selectfont \tiny \.{I}D}};
\end{tikzpicture}
\hspace{-2mm}
}
\foreach \x in {A, ..., Z}{%
\expandafter\xdef\csname orcid\x\endcsname{\noexpand\href{https://orcid.org/\csname orcidauthor\x\endcsname}{\noexpand\orcidicon}}
}

\def\BibTeX{{\rm B\kern-.05em{\sc i\kern-.025em b}\kern-.08em
    T\kern-.1667em\lower.7ex\hbox{E}\kern-.125emX}}
\begin{document}

\title{Instant AoI Optimization through Relay Location Selection in Disaster Multi-hop Communication}

\author{\IEEEauthorblockN{Yang Gao\hspace{-1.5mm}\orcidA{},
%~\IEEEmembership{Member, IEEE},
Zezhi Zeng\hspace{-1.5mm}\orcidB{}, Xiaoli Qin, Yifan Xu\hspace{-1.5mm}\orcidC{},
%~\IEEEmembership{Member, IEEE},\\
Haotian Wang\hspace{-1.5mm}\orcidD{}, and Jun Tao\hspace{-1.5mm}\orcidE{}}
\vspace{-20pt}
\thanks{Manuscript received XX XX, 2024.
This work was supported
in part by the CERNET Southeastern China (North) Regional Network Center,
in part by the Key Laboratory of Internet Information Retrieval of Hainan Province Research Fund under Grant 2023KY01,
in part by the Startup Fund of Hainan University under Grant KYQD(ZR)23058,
and in part by the Hainan Provincial Natural Science Foundation of China No.622RC617.
({\it Corresponding author: Yang Gao.})}
\IEEEcompsocitemizethanks{
\IEEEcompsocthanksitem{Yang Gao is
with the School of Cyberspace Security, Hainan University, Haikou, Hainan 570228, China,
and the Key Laboratory of CNII, MOE, China
(e-mail: yanggao@hainanu.edu.cn).}
\IEEEcompsocthanksitem{Zezhi Zeng and Xiaoli Qin are
with the School of Cyberspace Security, Hainan University, Haikou, Hainan 570228, China
(e-mail: \{zzz918, xiaoliqin\}@hainanu.edu.cn).}
\IEEEcompsocthanksitem{Yifan Xu, Haotian Wang and Jun Tao are
with the School of Cyber Science and Engineering, Southeast University, Nanjing, Jiangsu 211189, China
%\protect\\
(e-mail: \{xyf, haotianwang, juntao\}@seu.edu.cn).}
\IEEEcompsocthanksitem{Jun Tao is also
with the Key Laboratory of CNII, MOE, China,
%and the Purple Mountain Laboratories for Network and Communication Security, Nanjing, China.
and the Purple Mountain Laboratories, Nanjing, Jiangsu, China.}
}% <-this % stops a space
}

\maketitle
\begin{abstract}
%In a disaster, the infrastructure is destroyed, and the nodes lose the communication capabilities on which they rely.
%In such cases, the nodes can cooperate to store and forward messages through relays.
%In disasters, some rescue equipment (e.g., rescue robots, vehicle-mounted base stations, UAVs, mobile base stations, etc.) acts as relay nodes to support the entire network.
Meteorological disasters such as typhoons, forest fires, and floods can damage the communication infrastructures,
which will further disable the communication capabilities of cellular networks.
The multi-hop wireless communication based on IoT devices (e.g., rescue robots, UAVs, and mobile devices) becomes an available and rapidly deployable communication approach for search and rescue operations.
%Therefore, the preferences of relay nodes play a decisive role in network performance, specifically in terms of network delay and Age of Information (AoI), which is crucial for rescue operations.
However, Age of Information (AoI), an emerging network performance metric, has not been comprehensively investigated in this multi-hop model.
In this paper, we first construct a UAV-relayed wireless network model and formulate the end-to-end instant AoI.
Then we derive the optimal location of the relay UAV to achieve the minimum instant AoI
by mathematical analysis.
Simulations show that the derived relay location can always guarantee the optimal AoI and outperform other schemes.
%In this paper, we assume a multi-hop communication model where real-time data is transmitted via relaying with rescue equipment serving as relay nodes.
%Considering the limited information transmission capabilities of rescue equipment, the critical parameter of transmission distance greatly affects signal attenuation.
%We focus on the locations of devices that significantly impact network performance and strive to reduce network delay and improve AoI by adjusting the relay locations.
%Our conclusions are derived through rigorous mathematical reasoning.
%Finally, our inferences are validated through extensive simulation experiments, and the results demonstrate that the choice of relay location is of great significance for improving performance.
\end{abstract}

%-------------------------------------------------------
\section{Introduction}
According to United Nations' reports~\cite{disasterD},
there were more than $11,000$ reported disasters globally, with over two million deaths and \$$3.64$ trillion in losses.
Moreover, meteorological disasters such as typhoons, forest fires, and floods 
even damage the cellular network infrastructures (e.g., base stations)
and disable their capabilities,
which hinders the timely arrival of relief supplies or Search and Rescue operations.
For example, more than $12,500$ base stations in Hainan were disabled by Super Typhoon Yagi in $2024$~\cite{typhoonYagi}.
In the infrastructure-less scenarios, the multi-hop wireless communication can be regarded as an available and valuable solution,
%which will further disable the communication capabilities of cellular networks.
considering that many destination nodes may be located beyond the direct wireless communication coverages of source nodes.
%, with a typical application being Delay-Tolerant Networks (DTNs).
%DTNs have been considered a promising technology for supporting wireless communication, especially when wireless devices cannot be covered~\cite{Basu2020, Zguira2018, Nakagawa2020, Yaacoub2020}. 
%Due to the intermittent connections between nodes in DTNs, a crucial issue is how to improve the freshness of the information.
%-Similar to DTNs, multi-hop communication also needs to improve the freshness.

Age of Information (AoI)~\cite{Kaul2012, Yao2023} has been emerging as the most suitable performance metric
for measuring information freshness from the standpoint of Internet of Thing (IoT) devices.
Defined as the time elapsed since the generation time of the latest arrival packet at a destination node,
AoI characterizes the level of timely information delivery at the receiver side.
In contrast to conventional measures such as delay and throughput,
which capture the effectiveness of data collection and transmission in a network as a whole,
AoI aims to quantify the timeliness of updates from a destination's perspective.
{\it However, the AoI of multi-hop communication has not been investigated comprehensively.}

This paper focuses on the AoI optimization problem in a post-disaster multi-hop wireless communication model,
%where real-time data is transmitted via relaying with rescue equipment serving as relay nodes.
where the critical real-time data (e.g., videos, photos) needs to be transmitted from source to destination by communication relays.
In comparison with ground wireless relays (e.g., smartphones, laptops),
UAV-enabled communication relays: (1) can offer better wireless channels with line-of-sight (LoS) mode,
(2) are of fast deployment ability and flexible movability~\cite{Han23IEsJ}.
%% very easy to be deployed across the ruined roads.
%In disaster scenarios, the devices that typically exchange information are small mobile devices equipped with messaging functions (e.g., smartphones and mobile smart devices, etc.),
%which have extremely limited communication capabilities.
%To ensure the reliable construction of disaster communication, various mobile auxiliary IoT devices are widely utilized in disaster relief efforts,
%such as rescue robots~\cite{Shu2020}, vehicle-mounted base stations~\cite{Changle2019}, UAVs~\cite{Fan2018}, and backpack mobile stations.
%The addition of these rescue auxiliary devices is essential, as they are the cornerstone of building communication.
Since the locations of relays also affect the communication channel gain,
it is necessary to optimize the locations of UAV-enabled relays to maximize the end-to-end instant AoI.
%Moreover, the locations of UAV-enabled relays can affect the transmission performance, including the end-to-end instant AoI.
%
%However, the performance of these temporary auxiliary IoT devices remains significantly limited
%due to factors such as constrained battery life, limited signal transmission power, and restricted communication range.
%Considering the limited number of adjustable parameters, we naturally focus on optimizing the relay locations,
%which significantly impact data freshness, to improve the overall network performance.
The main contributions of this paper are as follows:
\begin{itemize}
    %\item We analyze the typical multi-hop communication model, use figures to show the continuous trend of AoI, and employ calculus to construct the mathematical expression.
    \item We construct the multi-hop wireless communication model, and formulate AoI by the calculus methodology.
%	\item We combine mathematical methods and logical inferences to analyze the impact of relay location on the system performance and derive the final results.
    \item We derive the analytical form of AoI with two situations separately, and obtain the minimum AoI and its corresponding relay location.
%    impact of relay's location on the system performance, and derive the final results.
%    \item Extensive simulation results demonstrate that the inferred results are highly consistent with actual scenarios, providing theoretical guidance for practical applications.
    \item We conduct the extensive simulations, which show that the theoretical analysis conforms to experimental results.
    %, providing theoretical guidance for practical applications.
\end{itemize}

The rest of the paper is organized as follows. Section \ref{sec:Related} reviews related work. Section \ref{sec:model} introduces the system model. Section \ref{sec:analysis} analyzes the system. Simulations are conducted in Section \ref{sec:Simulation} and Section \ref{sec:conclusion} concludes the paper.
%-------------------------------------------------------
\begin{figure*}
	\centering
	\includegraphics[width=0.65\textwidth]{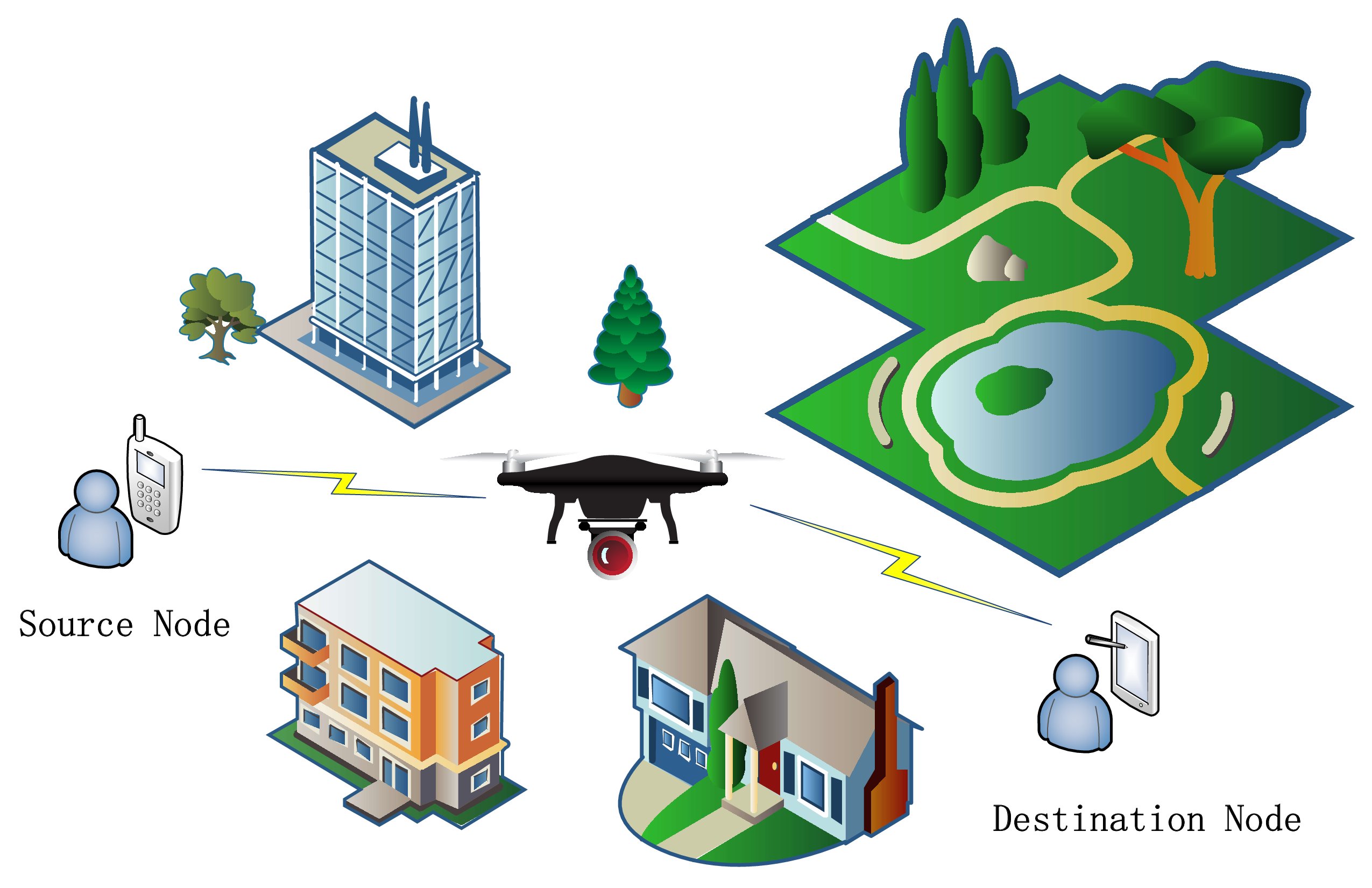}
	\caption{System model for multi-hop wireless communication}
	\label{fig:system}
    \vspace{-10pt}
\end{figure*}
%------------------------------------------------------- 
\section{Related Work}\label{sec:Related}
%-------------------------------------------------------
%In this section, we provide an overview of the related work in the areas of IoTs, Multi-hop Communication, and AoI, respectively.
This section provides an overview of the related works about IoTs, multi-hop communication, and AoI, respectively.

\subsection{Internet of Things (IoTs)}
The manufactured technologies of the IoT industry are becoming ubiquitous in modern civilization. The new developments in electronics and communication technology (\emph{RF semiconductors, system-on-chip microcontrollers, computing processors miniaturization}) enable flexible conceptualization and rapid implementation of numerous IoT applications. Additionally, the proliferation of IoT devices in recent years, ranging from body sensors and wearable devices to home appliances and industrial monitoring sensors, has driven the rapid evolution of various computing paradigms, such as pervasive and ubiquitous computing, mobile crowd sensing~\cite{Feng2021}, edge computing~\cite{Jing2023}, autonomous vehicles~\cite{Figlarz2022}, smart cities~\cite{Pati2024}, domotics, industrial IoT~\cite{Krishnan2022}, e-agriculture, advanced tracking and surveillance~\cite{Bera2022}, and intelligent public transportation. These technologies are now omnipresent across all sectors. The resulting networking and computing paradigms, applications, and services have expedited the rise of the big data era and significantly changed our daily lives.

\subsection{Multi-hop Communication}
If the source and the destination are far away from each other
or the channel between them is too weak for direct transmission,
multi-hop communication, where relays forward the source's data packets to the destination can be used.
DTNs are a typical application that relies on the movement of nodes to forward messages based on the routing mode of ``\emph{store-carry-forward}''~\cite{Fall2008}. Thus, DTNs are widely used in disaster relief and military communication networks~\cite{Stute2022, Stocchero2023}.~\cite{Peer2020} proposes a novel multi-hop D2D framework where user-BS (source-destination) pairing is optimized jointly with routing and scheduling to maximize the number of covered users in the disaster-affected area within a given time frame. 
~\cite{Ghosh2021} investigates a UAV-based cognitive hybrid multi-hop D2D communication using a clustering technique in a downlink scenario. The UAV can communicate through a multi-hop D2D network to reach the maximum number of victims.~\cite{Ma2022} presents an adaptive hybrid computation and multi-hop communication strategy to take full advantage of sensor processing capabilities and improve energy efficiency at the link level in the railway disaster wireless monitoring system, thereby improving the valid lifetime accordingly. In short, research on multi-hop communication routing is of great significance because it can improve the reliability and accessibility of data transmission in communication scenarios facing unstable, high-latency, or interrupted network connections.
\subsection{Age of Information (AoI)}
Extensive work has dealt with the theoretical foundations
of AoI, including packet generation rate control~\cite{Kaul2012, Sun2016}, 
queue management~\cite{Pappas2015}, and scheduling policy~\cite{Sastry2018}. Meanwhile, some practical settings have been taken into consideration when exploring AoI in real-world applications. They considered practical constraints including throughput~\cite{Kadota2018}, interference~\cite{Talak2020}, and channel access~\cite{Yates2017, Kuo2019}. Moreover,~\cite{Tripathi2023} has studied AoI under the setting where a mobile agent traverses ground terminals for data collection, whereas ~\cite{Yin2019} considered the request and response behaviors between users and the data distributor. Recent work extended AoI into IoTs and network edge, with~\cite{Zhou2018} aiming to minimize AoI under an average energy cost constraint at the IoT device and~\cite{Chengzhang2019} pursuing a general AoI model for the sampling behaviors at the network edge. Although existing related work has expanded the theoretical exploration of AoI into practical applications to some extent,
many important settings or application scenarios in IoTs remain insufficiently addressed.

This paper focuses on applying AoI to performance optimization in multi-hop relay communication, distinguishing itself from numerous previous research that solely addresses data delay and throughput, which reflect the effectiveness of data collection and transmission across the network. This is crucial for disaster rescue operations and saving lives.
%-------------------------------------------------------
\section{System Modeling}\label{sec:model}
%-------------------------------------------------------
%In this section, we first provide mathematical modeling for the communication model, network delay, and AoI.

\subsection{Communication Model}
Fig.~\ref{fig:system} illustrates a typical multi-hop communication model in a disaster, consisting of two nodes
(i.e., \emph{source node} and \emph{destination node}) and a UAV-enabled relay node.
Data packets can be continuously transmitted from the source node to the destination node
through this relay node.
%For ease of exposition, we will refer to source node, destination node, and relay node as ``\emph{node}'' and ``\emph{relay}'' in the rest of this paper, respectively.

Due to the similar features of the rescue equipment and the wide range of applications, we use a UAV as a typical relay node. Since the UAV is flying at a low altitude, we assume that all nodes are in horizontal communication mode. The UAV discussed in this paper is equipped with a wireless transceiver as another equipment, enabling communication with the nodes. The link rate between the UAV and a node adheres to
the principles of the Shannon-Hartley theorem~\cite{Shannon1948}. A concise formulation of this theorem is elucidated below:
\begin{equation}
		\begin{aligned}
			\label{eq1}
			C = B\log_2 (1+\gamma),
		\end{aligned}
\end{equation}
where $C$ denotes the link rate and $B$ denotes the communication bandwidth. Also, $\gamma$ stands for the signal-noise ratio (SNR). In this paper, the allocated bandwidth $B$ is fixed. Therefore, the link rate varies according to the SNR. The received signal power is denoted as $P_r$, which is determined by the Friis transmission equation \cite{Friis1946}, expressed as follows:
\begin{equation}
		\begin{aligned}
			\label{eq2}
			P_r=\left(\frac{\lambda}{4\pi d}\right)^2 \cdot G_tG_rP_t,
		\end{aligned}
\end{equation}
where $\lambda$ denotes the communication wavelength, $d$ denotes the transmission distance, $G_t$ and $G_r$ respectively denote antenna gain of transmitter and receiver, and $P_t$ denotes the transmission signal power. At this time, SNR is the received signal power divided by the noise power. Therefore, due to (\ref{eq2}), the SNR ($\lambda$) is defined as follows:
\begin{equation}
		\begin{aligned}
			\label{eq3}
			\gamma=\left(\frac{\lambda}{4\pi d}\right)^2 \cdot \left(\frac{G_tG_rP_t}{N}\right),
		\end{aligned}
\end{equation}
where $N$ denotes the noise power in the free space. Due to (\ref{eq3}), the link rate $C$ is greatly influenced by the transmission distance $d$.
\subsection{Network Delay}
Network delay typically refers to the time it takes for a data packet to travel from the source node to the destination node. The formula for network delay can vary based on the network environment and application scenario, but it generally includes key parameters such as transmission delay, queuing delay, processing delay, and propagation delay.
A common formula for network delay is:
\begin{equation}
		\begin{aligned}
			\label{delay}
			 D = D_t + D_q + D_p + D_{tr},
		\end{aligned}
\end{equation}
where $D$ denotes the total network delay.
$D_t$ is the transmission delay, which is the time required to send the data packet over the link. It is calculated as $D_t = \frac{L}{C}$, where $L$ is the length of the packet (\emph{in bits}) and $C$ is the transmission rate (\emph{in bits per second}). $D_q$ is the queuing delay, which is the time the packet spends waiting in the router or switch queue before being processed. This depends on network traffic and the processing capacity of the devices.
$D_p$ is the processing delay, which is the time required by the router or switch to process the packet, including checking the packet header and determining the forwarding path. $D_{tr}$ is the propagation delay, which is the time it takes for the packet to propagate through the physical medium. It is calculated as $D_{tr} = \frac{d}{s}$, where $d$ is the distance of the link (\emph{in meters}) and $s$ is the propagation speed of the signal in the medium (\emph{typically close to the speed of light, approximately $3 \times 10^8$ meters/second}).

In this paper, the data packets we need to transmit are relatively small. We assume that the packets do not need to queue at the nodes, and both processing delay and propagation delay can be ignored. Therefore, the total network delay is equal to the transmission delay, i.e., $D = D_t$.

\subsection{AoI Formulation}
In many cases, outdated
information can lead to inaccurate decision-making, while
timely information can provide superior support and guidance, especially in emergencies. Consequently, ensuring the freshness of information is of paramount importance in numerous domains.
%\begin{figure}[ht!]
\begin{figure}
	\centering
	\includegraphics[width=0.95\columnwidth]{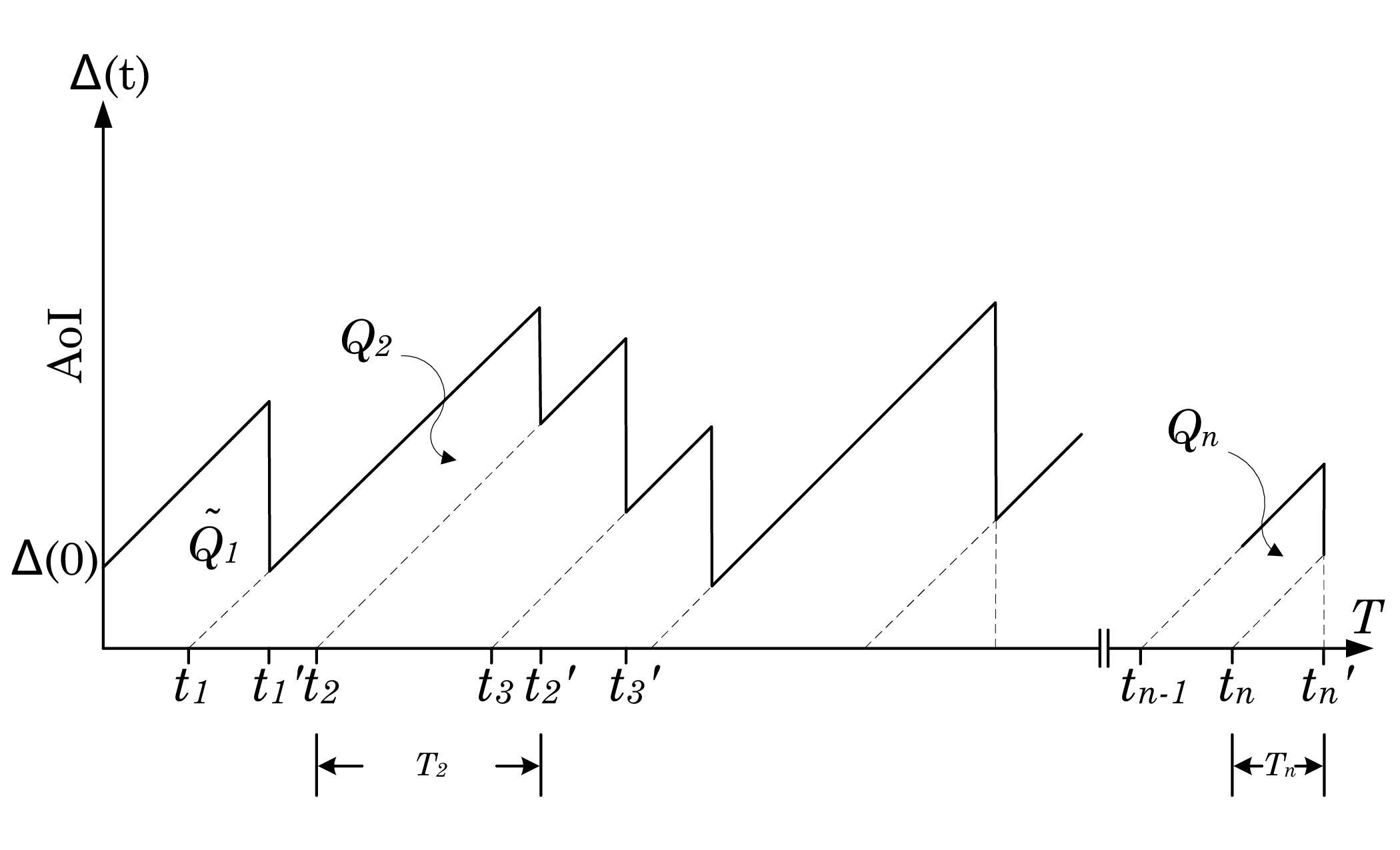}
	\caption{An example of varying AoI at the receiving end}
	\label{fig:aoi}
    \vspace{-5pt}
\end{figure}

In this paper, information freshness is also the primary metric we aim to achieve.
To measure data freshness and capture the latest updates, we adopt the recently proposed Age of Information (AoI) metric, which evaluates data freshness at data nodes. As depicted in Fig.~\ref{fig:aoi}, a sample variation of age $\Delta(t)$ as a function of time $t$ at the receiving end is shown. A concise formulation of AoI is explained below:
\begin{equation}
		\begin{aligned}
			\label{aoi}
			 \Delta(t)=t-\mu(t),
		\end{aligned}
\end{equation}
%-------------------------------------------------------
where $\mu(t)$ and $\Delta(t)$ denote the generation time of the latest packet from a source node and its AoI at the destination node at time $t$, respectively.
Over an interval $(0,T)$, the average AoI (\cite{Kaul2012}) can be expressed as
\begin{equation}
		\begin{aligned}
  \Delta(T)=\frac{1}{T}\int_{0}^{T} \Delta(t) dt.
        \label{aoi_integral}
		\end{aligned}
\end{equation}

The first status update is generated at $t_1$, followed by updates at $t_2$, $t_3$, $\dots$, $t_n$. The age at the receiving end increases linearly over time in the absence of any updates and is reset to a smaller value when an update is received. Update $i$, generated at time $t_i$ , finishes service and is received by the receiving end at time $t_i'$ . At $t_i'$, the age $\Delta(t_i')$ at the receiving end is reset to the age $T_i = t_i'-t_i$ of the received status update. The age $T_i$ is also the system time of the update packet $i$ and is the sum of the time the packet waited in the queue and the time it spent in service. Thus the age function $\Delta(t)$ exhibits the sawtooth pattern shown in Fig.~\ref{fig:aoi}. The time average age of the status updates is the area under the sawtooth function in Fig.~\ref{fig:aoi} normalized by the time interval of observation.

For simplicity of exposition, the length of the observation interval is chosen to be $T = t_n'$, as depicted in Fig.~\ref{fig:aoi}. We decompose the area defined by the integral in~(\ref{aoi_integral}) into a sum of disjoint geometric parts. Starting from $t = 0$, the area can be viewed as the concatenation of the polygon area $\wt{Q}_1$, the trapezoids $Q_i$ for $i \geq 2$ (\emph{with $Q_2$ and $Q_n$ highlighted in the figure}), and the triangular area of width $T_n$ over the time interval ($t_n$, $t_n'$). With $N(T)$ = $\max\{n | t_n \leq T \}$ denoting the number of arrivals by time $T$, this decomposition yields: 
\begin{equation}
		\begin{aligned}
  \Delta(T)=\frac{\wt{Q}_1+T_n^2/2+\sum_{i=2}^{N(T)}Q_i}{T}. 
        \label{aoi_1}
		\end{aligned}
\end{equation}
\section {Mathematical Analysis}\label{sec:analysis}
%\begin{figure}[ht!]
\begin{figure}
	\centering
	\includegraphics[width=0.8\columnwidth]{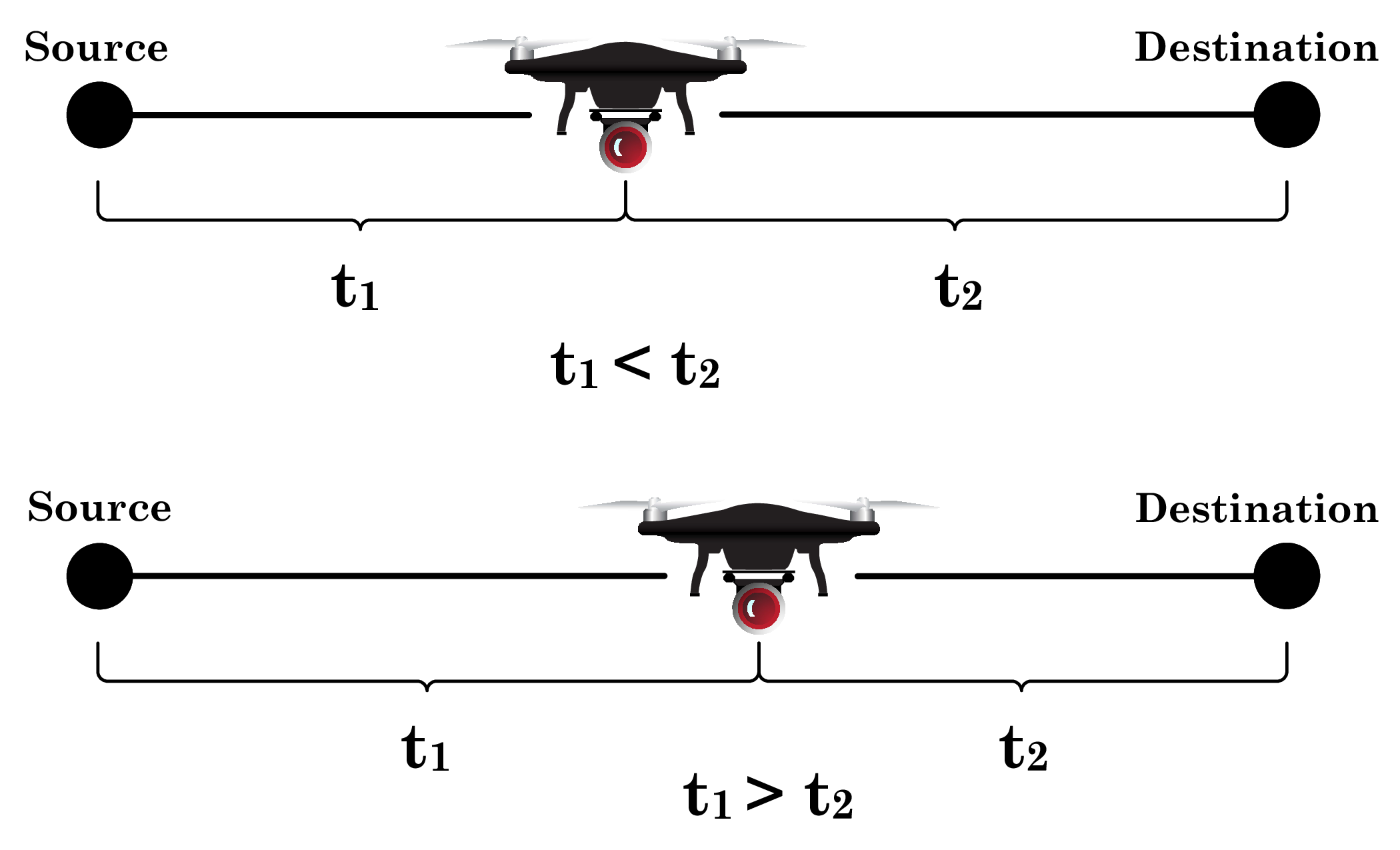}
	\caption{Typical System Model for multi-hop Communication.}
	\label{fig:System_Model}
\end{figure}
In this section, we provide the mathematical analysis in the areas of network delay, and AoI.
We generalize the typical multi-hop relay communication model, as shown in Fig.\ref{fig:System_Model}.
Assume that there are $n$ data packets continuously sent from source to destination via the UAV-enabled relay.
Each packet needs to wait for the previous data packet to be sent before it can be transmitted.
Here the relay can continue to forward packets without waiting, operating in full duplex mode.

\subsection{Network Delay}
First, we try to analyze the network delay under different conditions,
according to the moment the first packet reaches the destination node,
as shown in Fig.~\ref{fig:1st}.
Here $t_1$ denotes the time taken to transmit a packet from the source to the relay,
and $t_2$ denotes the time taken to transmit a packet from the relay to the destination.
%-------------------------------------------------------
%-------------------------------------------------------
%\begin{figure}[ht!]
\begin{figure}
	\centering
	\includegraphics[width=0.9\columnwidth]{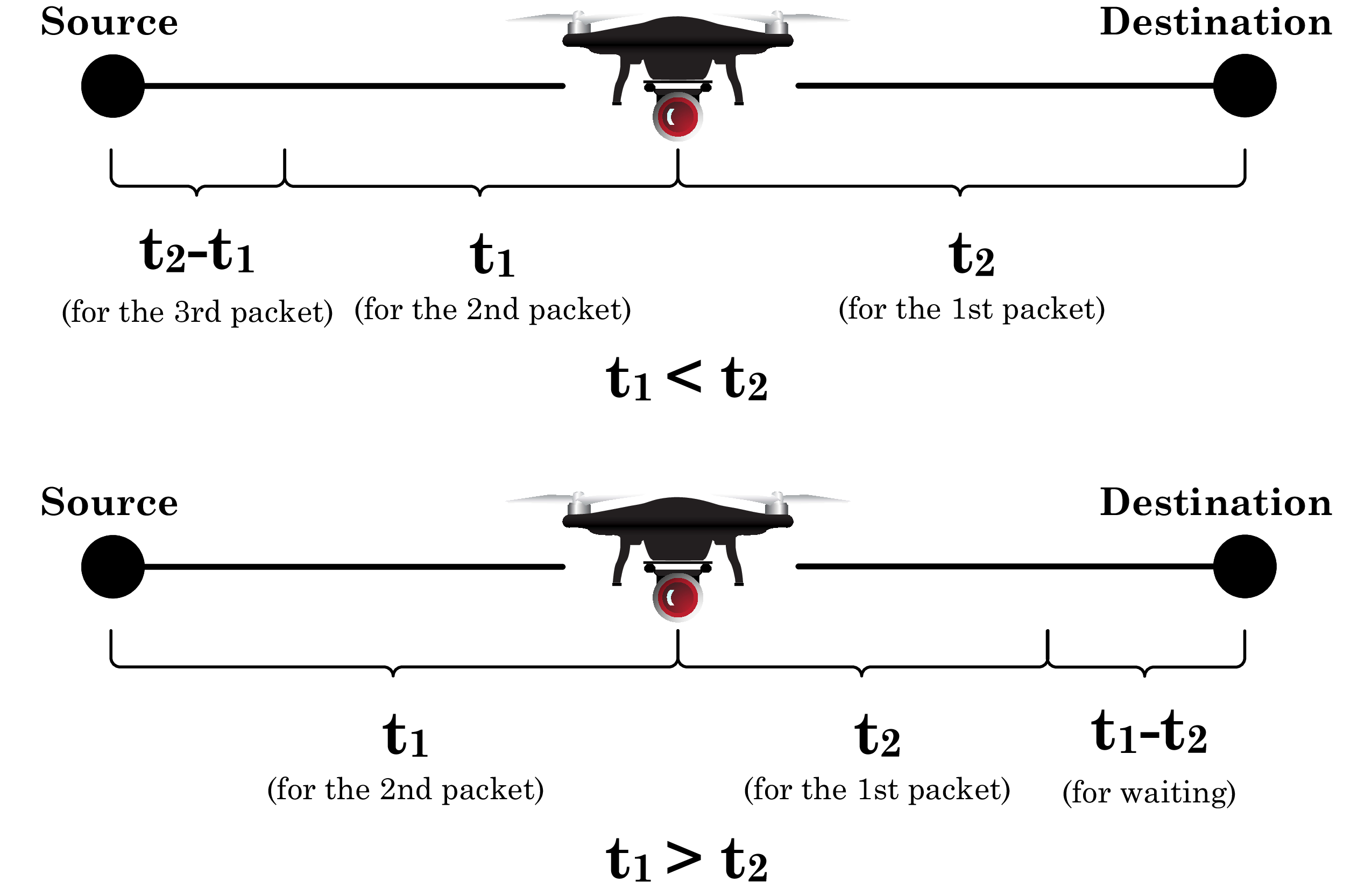}
	\caption{The 1st packet reaches the destination node.}
	\label{fig:1st}
\end{figure}
%-------------------------------------------------------
\subsubsection{Case $1$: $t_1<t_2$}
In the case where $t_1<t_2$, after the first packet arrives at the destination node at time ``\emph{$t_1+t_2$}'', the second packet is already waiting at the relay. This is because when the first packet starts transferring from the relay to the destination node at time ``\emph{$t_1$}'', the second packet begins transmission. Given that $t_1<t_2$, the next packet will arrive at its next node (\emph{relay node}) faster than the previous one. Following this logic, the average network delay can be:
\begin{equation}
\begin{aligned}
    \ol{D}_{t_1<t_2} &= \lim_{n\rightarrow\infty}\frac{t_1+t_2+(n-1)t_2}{n}\\
    &= \lim_{n\rightarrow\infty}\frac{t_1+nt_2}{n}\\
    &= t_2.
    \label{delay_1}
\end{aligned}
\end{equation}

\subsubsection{Case $2$: $t_1>t_2$}
In the case where $t_1>t_2$, the source node continuously sends packets one by one. Once a packet arrives at the relay, it will be transmitted to the destination node without any waiting time, until the last packet is transmitted. This is because the previous packet always finishes its transmission before the next packet is ready at the relay (\emph{$t_1>t_2$}).
Following this similar analysis, we can also derive the average network delay for {\it Case $2$}, i.e.,
\begin{equation}
\begin{aligned}
    \ol{D}_{t_1>t_2} &= \lim_{n\rightarrow\infty}\frac{nt_1+t_2}{n} = t_1. 
    \label{delay_2}
\end{aligned}
\end{equation}

Our goal is to minimize network delay, i.e., $\min \ol{D}=\min\{\max \{t_1,t_2 \}\}$.
It is obvious that $min\ol{D}=t_1=t_2$ is obtained only when $t_1=t_2$. 

\subsection{Instant AoI}
To measure the freshness of data and capture the latest data updates,
we adopt the average instantaneous AoI at the destination node at time $t$.
This paper assumes that data packets are continuously sent from the source node starting at time zero,
where AoI follows a repetitive change pattern.
Unlike the classical calculation method, we use geometric methods combined with definite integrals for the calculation. 

In the case where $t_1<t_2$, a sample variation of AoI is shown in Fig.~\ref{fig:AoI_t2}.
Based on the definite integrals, the average AoI can be calculated as follows.
\begin{equation}
\begin{aligned}
&\!\Delta(T)_{t_1<t_2} \\
=& \frac{1}{T}\int_{0}^{T} \Delta(t) dt \\
=& \lim_{n\rightarrow\infty}\frac{1}{t_1+nt_2}(\int_{0}^{t_1+2t_2} t dt+(n\!-\!2)\!\int_{t_1+2t_2}^{t_1+3t_2} t-t_1 dt\\
&+ t_2(t_2-t_1)+2t_2(t_2-t_1)+\cdots+(n-3)t_2(t_2-t_1))\\
=& \frac{5}{2}t_2+\frac{n}{2}(t_2-t_1).
\label{aoi_2}
\end{aligned}
\end{equation}
%-------------------------------------------------------
%\begin{figure}[ht!]
\begin{figure}
	\centering
	\includegraphics[width=0.9\columnwidth]{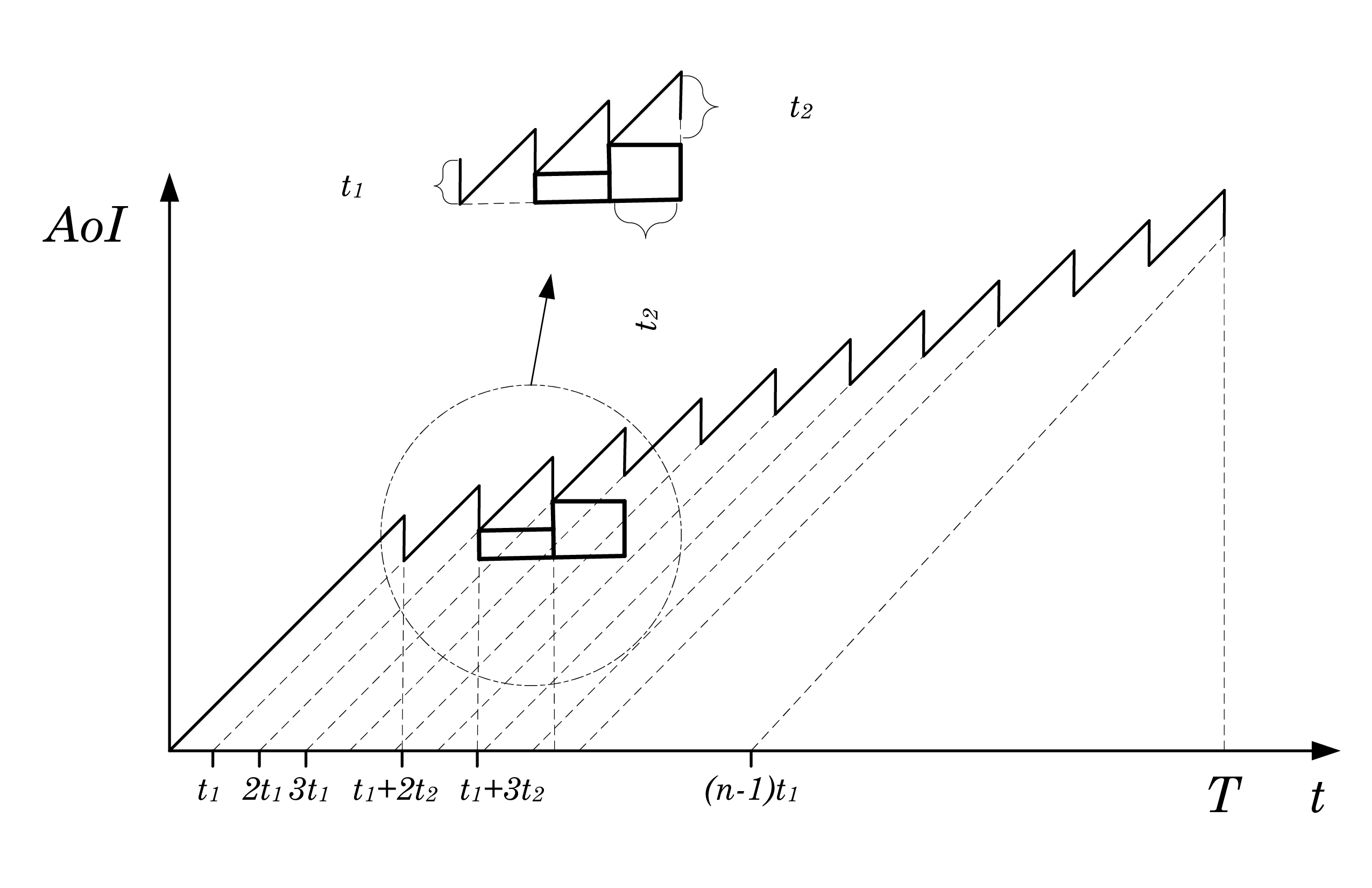}
	\caption{AoI respect to $t_1<t_2$.}
	\label{fig:AoI_t2}
\vspace{-10pt}
\end{figure}

It is important to point out that the first part of the integrals corresponds to integrating the first large triangle in Fig.~\ref{fig:AoI_t2}. The second part of the integrals integrates the ``\emph{$n-2$}'' trapezoids that represent the growth amounts in Fig.~\ref{fig:AoI_t2}. The third part of the integral accumulates the ``\emph{$n-3$}'' growth amounts, and the difference between the next growth amount and the previous one is $t_2(t_2-t_1)$, depicted as the rectangle in Fig.~\ref{fig:AoI_t2}, making this an increasing function.

In the case where $t_1>t_2$, a sample variation of AoI is shown in Fig.~\ref{fig:AoI_t1}. The average AoI, using definite integrals, can be calculated as follows:
\begin{equation}
\begin{aligned}
&\!\Delta(T)_{t_1>t_2}\\
=& \frac{1}{T}\int_{0}^{T} \Delta(t) dt \\
=& \lim_{n\rightarrow\infty}\frac{1}{nt_1+t_2}(\int_{0}^{2t_1+t_2} t dt+(n\!-\!2)\!\int_{2t_1+t_2}^{3t_1+t_2} t-t_1 dt)\\
=& \frac{3}{2}t_1+t_2,
\label{aoi_3}
\end{aligned}
\end{equation}
similarly to case one, the first part of the integrals integrates the first large triangle in Fig.~\ref{fig:AoI_t1}. The second part of the integrals integrates the ``\emph{$n-2$}'' trapezoids of the same shape.

The main theorem in this paper can be stated as follows.
\begin{thm}[$t_1=t_2\in\argmin_{t_1,t_2} \Delta(T)$, concerning $P_{node}<P_{relay}$, under the same noise $N_1=N_2=N$, where $N_1$ is the noise between the source and the relay, and $N_2$ is the noise between the relay and the destination]
\label{thm:periodic}
When the transmission signal power of the relay is greater than that of the user nodes (i.e, source node, destination node), the average instantaneous AoI of the receiver reaches its minimum when $t_1=t_2$, is
\begin{equation}
		\begin{aligned}
         \min \Delta(T) = \frac{5}{2}t_1= \frac{5}{2}t_2.
        \label{aoi_4}
		\end{aligned}
\end{equation}
\end{thm}
\begin{proof} 
In the model, the distance $d$ is fixed, and the length of the data packet is $L$. By combining 
\eqref{eq1}, \eqref{eq2}, and \eqref{eq3}, we have: 
\begin{equation}
		\begin{aligned}
			\label{eqn:t1}
			t_1 = \frac{L}{B\log_2(1+(\frac{\lambda}{4\pi})^2\frac{G_t\ol{G_r}}{N}\frac{P_{node}}{d_1^2})},
		\end{aligned}
\end{equation}
\begin{equation}
		\begin{aligned}
			\label{eqn:t2}
			t_2 = \frac{L}{B\log_2(1+(\frac{\lambda}{4\pi})^2\frac{\ol{G_t}G_r}{N}\frac{P_{relay}}{d_2^2})},
		\end{aligned}
\end{equation}
\begin{equation}
		\begin{aligned}
			\label{eqn:d}
			d = d_1+d_2=\frac{\gamma}{4\pi}(\sqrt{\frac{\frac{\ol{G_t}G_rP_{relay}}{N}}{2^{\frac{L}{Bt_2}}-1}}+\sqrt{\frac{\frac{G_t\ol{G_r}P_{node}}{N}}{2^{\frac{L}{Bt_1}}-1}}).
		\end{aligned}
\end{equation}

Where $d_1$ is the distance between the source node and the relay, and $d_2$ is the distance between the relay and the destination node. $G_t$ and $G_r$ respectively denote the antenna gains of the transmitter and receiver at the node, while $\ol{G_t}$ and $\ol{G_r}$ respectively denote the antenna gains of the transmitter and receiver at the relay. We assume that $G_t=G_r$, and $\ol{G_t}=\ol{G_r}$, so $\ol{G_t}G_r=G_t\ol{G_r}$ is always true. When $t_2\uparrow$, $t_1\downarrow$ is always true. When $t_1<t_2$, as $t_2$ increases, (\ref{aoi_2}) increases linearly. Once $t_1=t_2$ (i.e., $\sqrt{\frac{P_{node}}{P_{relay}}}=\frac{d_1}{d_2}$), $\min \Delta(T) = \frac{5}{2}t_2$ can be obtained. When $t_1>t_2$, the changing trends of $t_1$ and $t_2$ are analyzed in Appendix A.
\end{proof}
%\begin{figure}[ht!]
\begin{figure}
	\centering
	\includegraphics[width=0.9\columnwidth]{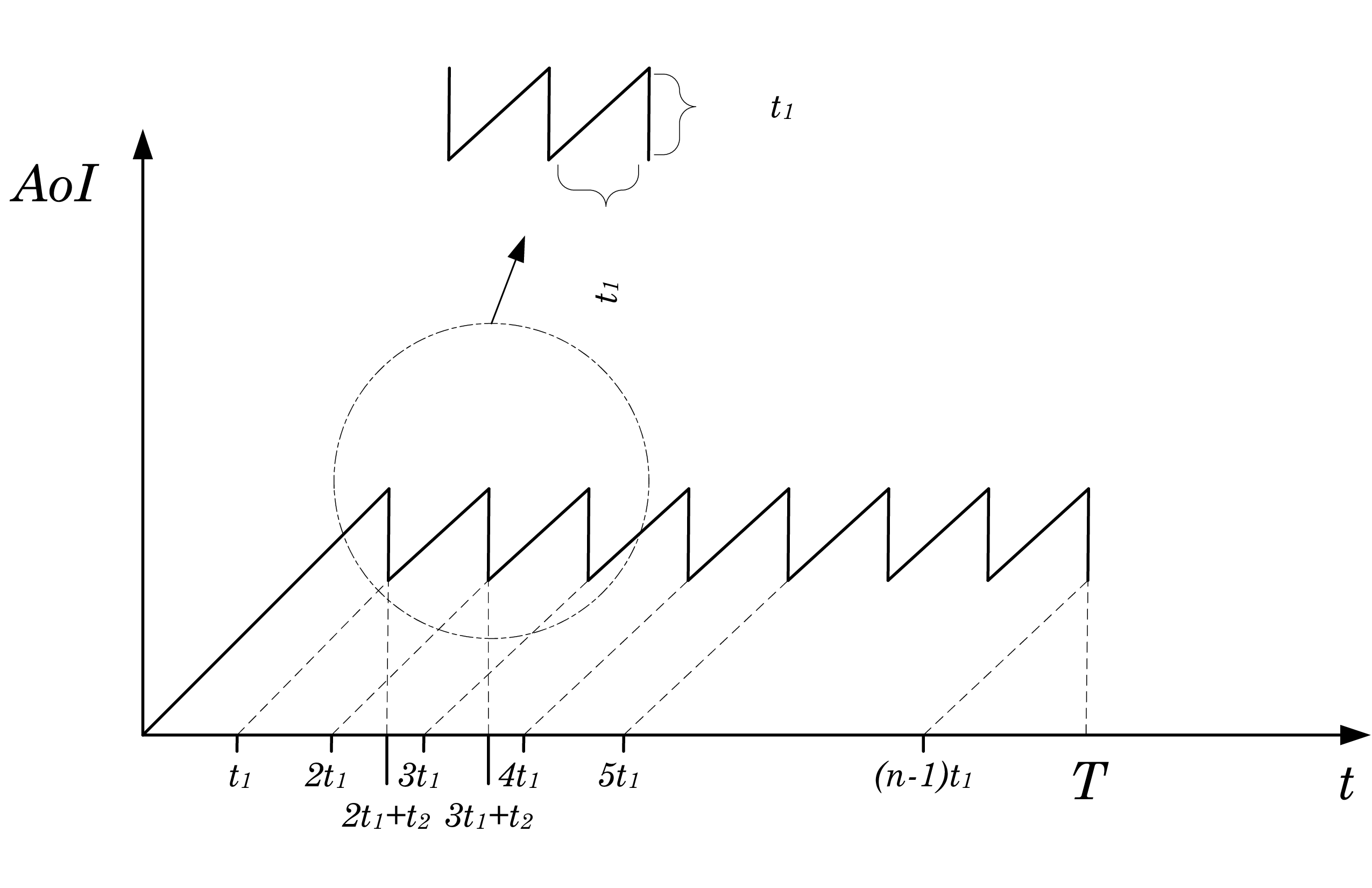}
	\caption{AoI respect to $t_2<t_1$}
	\label{fig:AoI_t1}
\vspace{-10pt}
\end{figure}
%-------------------------------------------------------

%\begin{rem}{\em
%Theorem~\ref{thm:periodic} reveals that when $t_1=t_2$, this point is both the minimum point of average network delay and the local minimum point of instantaneous AoI. It is also proved that under given conditions ($P_{relay}>P_{node}$), it is highly likely to be the global minimum point. This is of great significance for optimizing the performance of the entire system. Even if it is not the global minimum point of AoI, the true global minimum point of AoI would deviate from the optimal point of network delay, thus diminishing its practical significance, as network delay is also a critical factor in practice.}
%\end{rem}

\textbf{Remark}: Theorem~\ref{thm:periodic} reveals that when $t_1=t_2$, this point is both the minimum point of average network delay and the local minimum point of instantaneous AoI.
It is also proved that under given conditions ($P_{relay}>P_{node}$, {\it rescue equipment usually has relatively strong communication capabilities, which is reasonable and practical}), it is highly likely to be the global minimum point.
This is of great significance for optimizing the performance of the entire system.
Even if it is not the global minimum point of AoI, the true global minimum point of AoI would deviate from the optimal point of network delay,
thus diminishing its practical significance. Network delay is also a very critical factor in practice that cannot be ignored.

\subsection{Effect of Noise on the Optimal Location}
In this paper, we initially assume that the noise level is uniform across each link,
where the signal intensity of the white noise and the noise power are both set to $1\times 10^{-10}$.
Subsequently, we analyze the scenarios where the noise levels differ between links.
%-------------------------------------------------------
\begin{figure*}
    \centering
    \begin{subfigure}{0.33\textwidth}
    	\centering
    	\includegraphics[width=\linewidth]{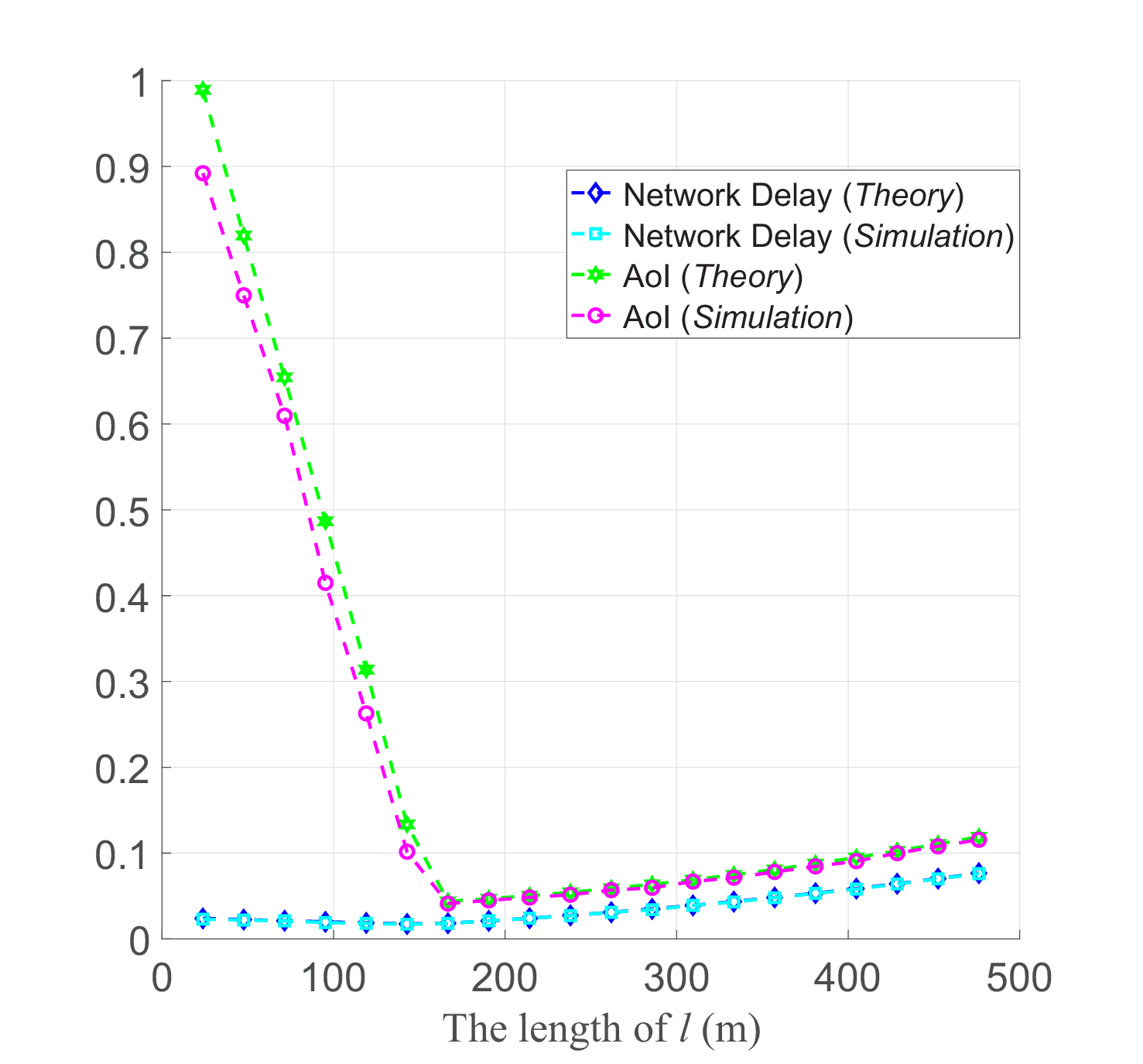}
    	\caption{Respect to $P_{node}=0.1, P_{relay}=0.5, d=500$.}
    	\label{sim:fig5a}
    \end{subfigure}%
    \begin{subfigure}{0.33\textwidth}
    	\centering
    	\includegraphics[width=\linewidth]{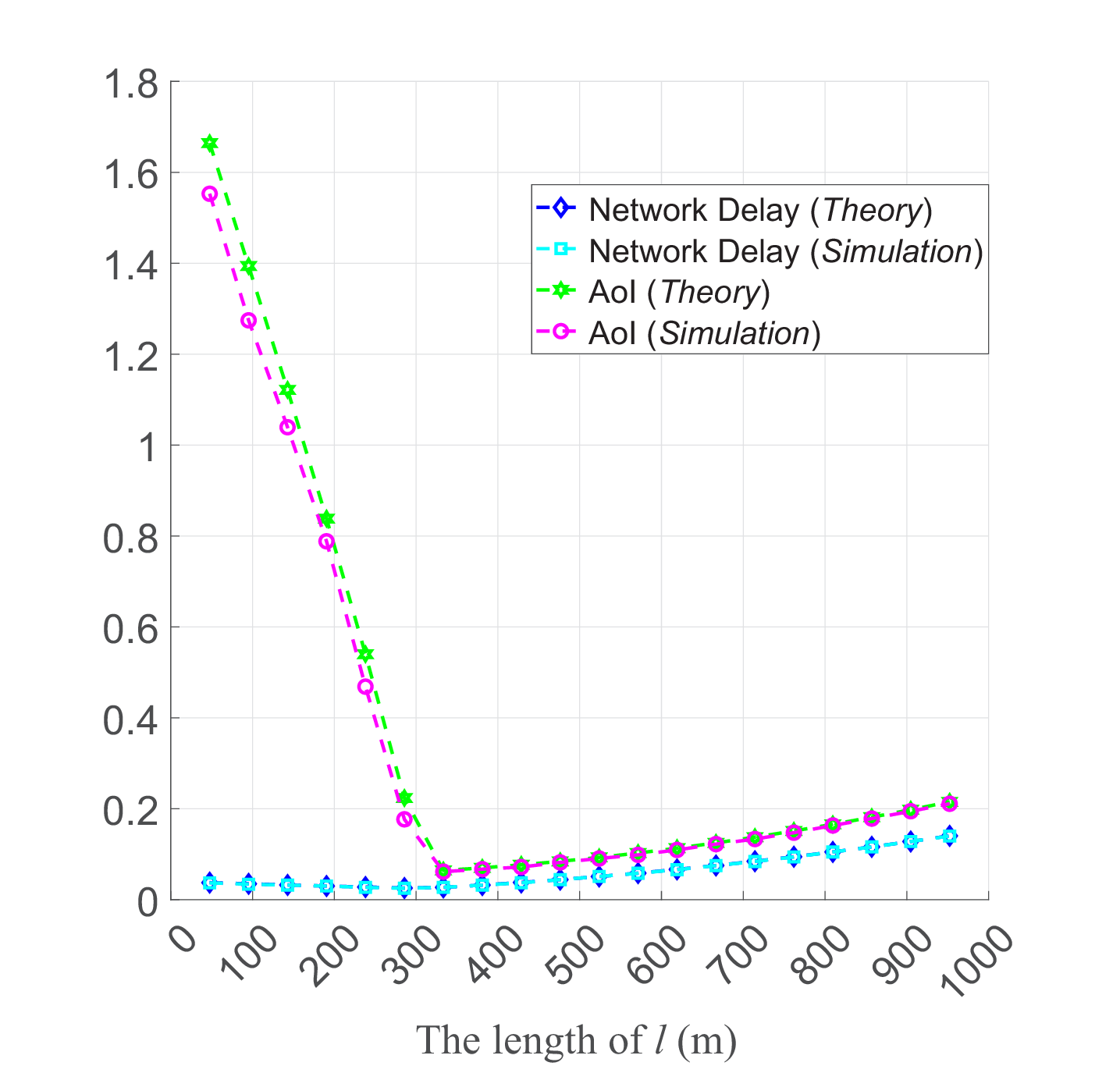}
    	\caption{Respect to $P_{node}=0.2, P_{relay}=1, d=1000$.}
    	\label{sim:fig5b}
    \end{subfigure}%
    \begin{subfigure}{0.33\textwidth}
    	\centering
    	\includegraphics[width=\linewidth]{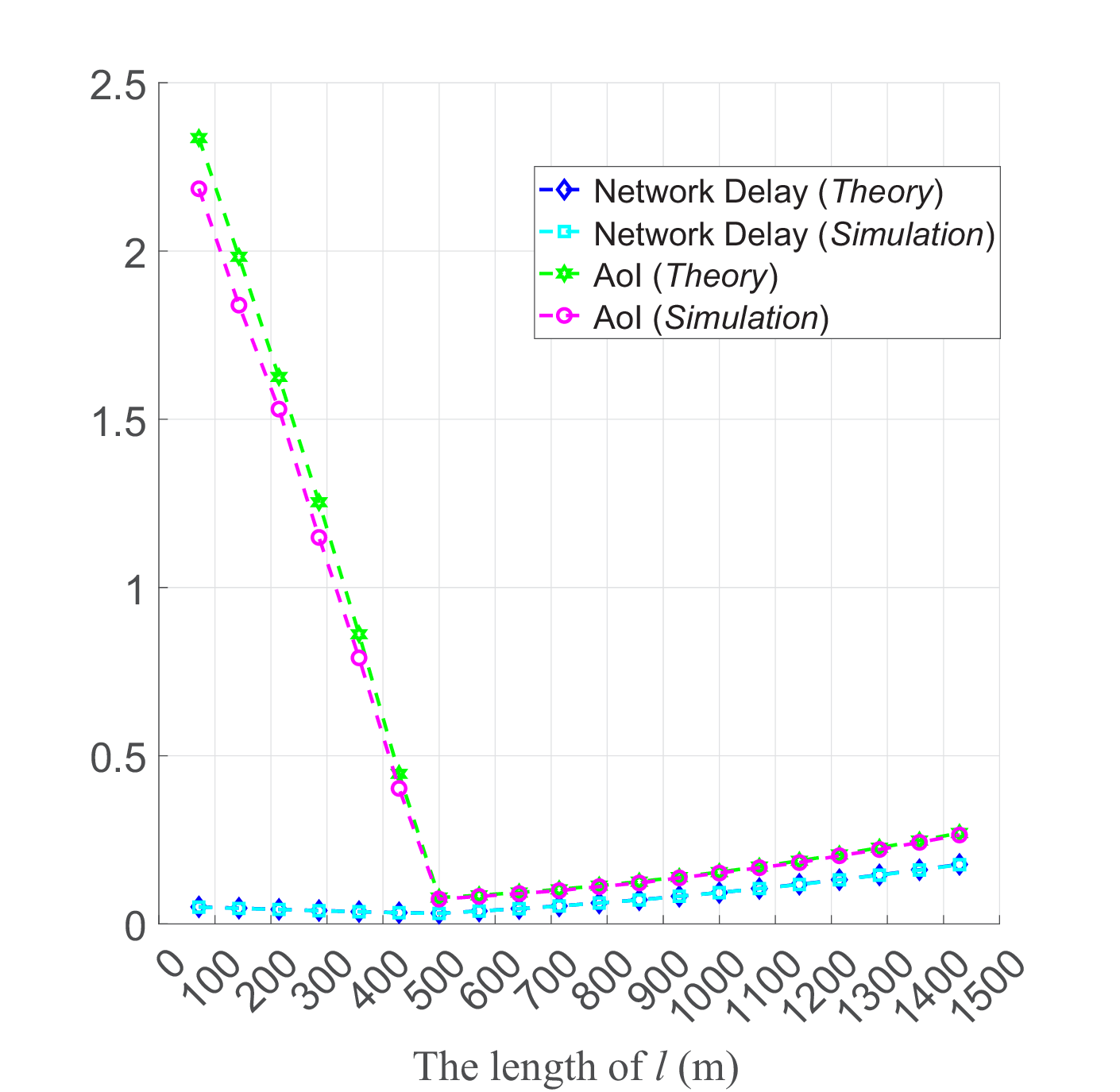}
    	\caption{Respect to $P_{node}=0.35, P_{relay}=1.5, d=1500$.}
    	\label{sim:fig5c}
    \end{subfigure}%
%    \caption{Average delay and Instant AoI respect to different $P_{node}, P_{relay}, d, l$.}
    \caption{Average delay and Instant AoI with varying $l$.}
    \label{sim:aoi1}
\vspace{-10pt}
\end{figure*}

An optimal location for the relay is defined in this paper as the location that minimizes the lower bound of AoI.
Based on our previous derivation, the optimal location is the point that satisfies $t_1=t_2$.
%Therefore, in order to determine the optimal location,  should solve the following equation:
Therefore, the optimal location of the UAV-enabled relay should satisfy:
\begin{equation}
%\nonumber
\begin{aligned}
	\frac{L}{B\log_2(1+(\frac{\lambda}{4\pi})^2\frac{G_t\ol{G_r}}{N}\frac{P_{node}}{l^2})}=\\
\frac{L}{B\log_2(1+(\frac{\lambda}{4\pi})^2\frac{\ol{G_t}G_r}{N}\frac{P_{relay}}{(d-l)^2})},
\end{aligned}
\end{equation}
which means that
\begin{equation}
	\begin{aligned}
	\label{eqn:t1=t2:final2}
		\frac{P_{node}}{l^2}= \frac{P_{relay}}{(d-l)^2}.
	\end{aligned}
\end{equation}
Based on Eq. (\ref{eqn:t1=t2:final2}), the optimal location of the relay when $N_1=N_2=N$ can be obtained, i.e.,
\begin{equation}
		\begin{aligned}
		\label{eqn:t1=t2:final3}
			l=\frac{d}{1+\sqrt{\beta}},
		\end{aligned}
\end{equation}
where $\beta=P_{relay}/P_{node}>1$.
Here the optimal location of the relay can minimize the lower bound of AoI.

Moreover, we can derive the optimal location when $N_1 \neq N_2$.
Similarly, the following conditions should be satisfied:
\begin{equation}
	\begin{aligned}
		\frac{L}{B\log_2(1+(\frac{\lambda}{4\pi})^2\frac{G_t\ol{G_r}}{N_1}\frac{P_{node}}{l^2})}= \\
    \frac{L}{B\log_2(1+(\frac{\lambda}{4\pi})^2\frac{\ol{G_t}G_r}{N_2}\frac{P_{relay}}{(d-l)^2})},
	\end{aligned}
\end{equation}
which means that
\begin{equation}
	\begin{aligned}
	\label{eqn:t1=t2:final4}
		\frac{P_{node}}{N_1l^2}= \frac{P_{relay}}{N_2(d-l)^2}.
	\end{aligned}
\end{equation}
Thus we have
\begin{equation}
		\begin{aligned}
		\label{eqn:t1=t2:final5}
			l=\frac{d}{1+\sqrt{\beta \frac{N_1}{N_2}}}.
		\end{aligned}
\end{equation}

In the case of $N_1 > N_2$, compared with the value in (\ref{eqn:noeq3}), $\frac{d}{1+\sqrt{\beta \frac{N_1}{N_2}}}<\frac{d}{1+(\frac{2}{3}\beta \frac{N_1}{N_2})^\frac{1}{3}}$ is always true, and combined with the inference in (\ref{eqn:leq}), the final result is also valid.

In the case of $N_1 < N_2$, $\frac{d}{1+\sqrt{\beta \frac{N_1}{N_2}}}<\frac{d}{1+(\frac{2}{3}\beta \frac{N_1}{N_2})^\frac{1}{3}}$ is not always true,
but holds only under certain restrictions (with appropriate parameters, e.g., $\beta$ as large as possible, and $\frac{N_1}{N_2}$ as close to ``\emph{1}'' as possible),
which we will show in the simulation in the next section.

\section{Performance Evaluation}\label{sec:Simulation}
%\begin{table}[ht]
In this section, we employ simulations to validate the optimization of network delay and AoI by changing the locations of the relays in a multi-hop relay communication model in MATLAB.\footnote{The codes are available at \href{https://github.com/hbkhhh619315/relaylocation}{https://github.com/hbkhhh619315/relaylocation}.}.
\begin{table}
\captionsetup{font=footnotesize}
\caption{Notations and Values}
\centering
\resizebox{\columnwidth}{!}{%
	\begin{tabular}{|c|p{0.2\linewidth}|} 
		\hline
		\textbf{Notations} & \textbf{Values} \\ [0.5ex] 
		\hline 
            The length of data packet (bit) ($L$) & $800000$ \\
        \hline
            The number of data packets & $100$ \\
        \hline
            Bandwidth in megahertz (MHz) ($B$) & $20$ \\
        \hline
			Wavelength in meter (m) ($ \lambda$) & $0.125$ \\
        \hline
%		Transmission signal power in Watt (W) ($P_t$) & $0.25$\\
%       \hline
		Noise power in Watt (W) ($N$) & $1 \times 10^{-10}$ \\
        \hline
        	Antenna gain of transmitter of nodes ($G_t$) & $0.9$ \\
        \hline
        	Antenna gain of receiver of nodes ($G_r$) & $0.9$ \\
        \hline
        	Antenna gain of transmitter of relay ($\ol{G_t}$) & $1$ \\
        \hline
        	Antenna gain of receiver of relay ($\ol{G_r}$) & $1$ \\
        \hline
\end{tabular}}
\vspace{-10pt}
\label{table:sysParams}
\end{table}
\subsection{Evaluation Environment}
Our simulation setup involves a system comprising two nodes and one relay. The source node transmits $100$ data packets ($800,000$ bits per packet) to the destination node through the relay. The transmission signal power of the nodes belongs to the finite set $\{0.1,0.2,0.35\}$ or the intervals $(0.05,0.35)$, $(0.1,0.7)$, $(0.15,1.05)$, while the transmission signal power of the relay belongs to the finite set $\{0.5,1,1.5\}$. The distance between the two nodes belongs to the finite set $\{500m,1000m,1500m\}$. The length $l$ ranges in $(d/5,t_1=t_2,d/2,4d/5)$. 
In the case of $N_1>N_2$, we set $N_1=2\times 10^{-10}, N_2=1\times 10^{-10}$. In the case of $N_1<N_2$, we set $N_1=1\times 10^{-10}, N_2=2\times 10^{-10}$. The other simulation parameter values were derived from reference~\cite{Horiuchi2016}. 
For more detailed parameter information, please refer to the original paper or consult Table~\ref{table:sysParams}.
%\begin{figure*}[ht!]
%\begin{figure*}[ht!]
\begin{figure*}
    \centering
    \begin{subfigure}{0.33\textwidth}
    	\centering
    	\includegraphics[width=\linewidth]{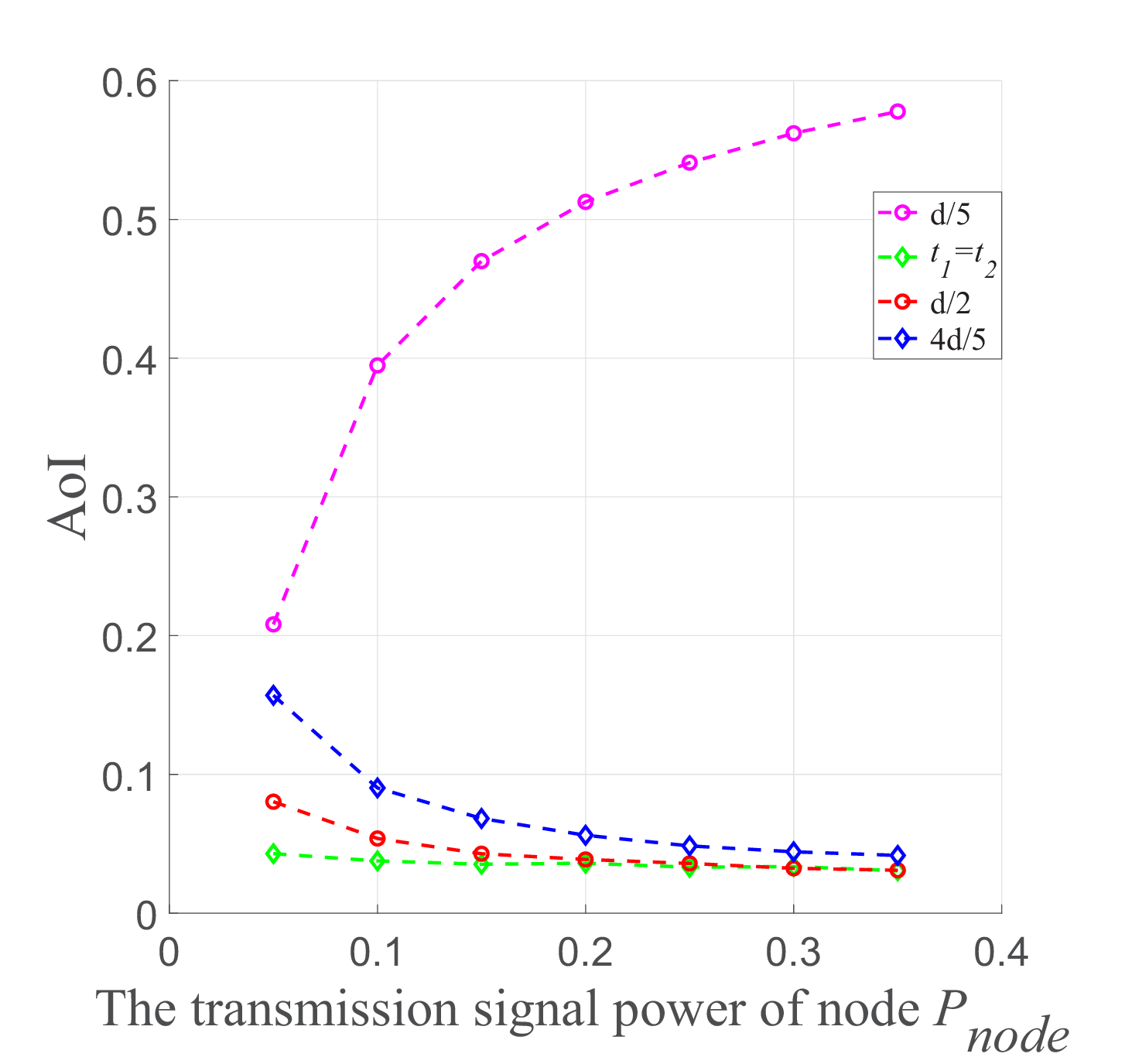}
    	\caption{Respect to different $P_{node}, P_{relay}=0.5, d=500$, $l$ ranges in $(d/5,t_1=t_2,d/2,4d/5)$.}
    	\label{sim:fig6a}
    \end{subfigure}%
    \begin{subfigure}{0.33\textwidth}
    	\centering
    	\includegraphics[width=\linewidth]{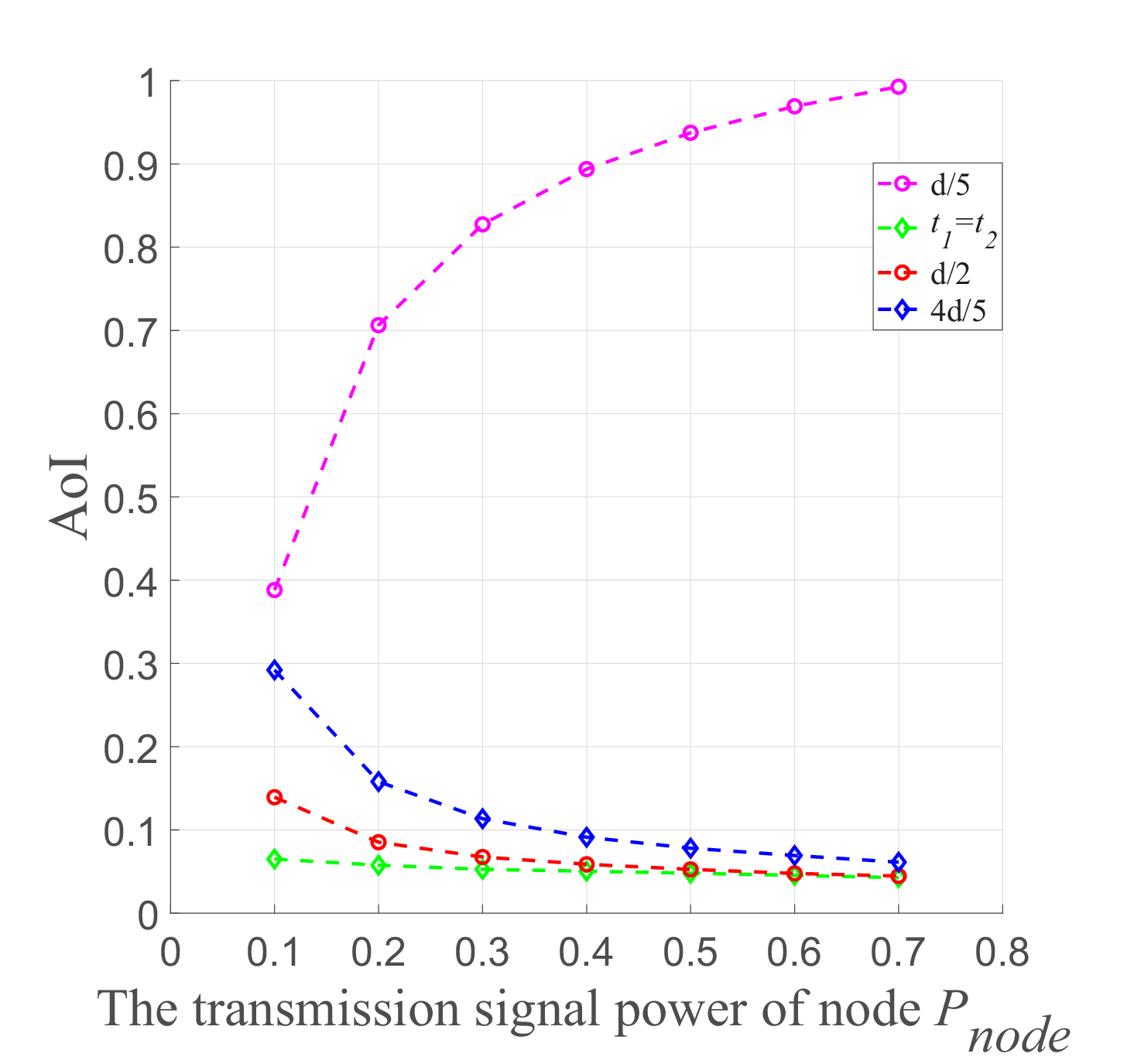}
    	\caption{Respect to different $P_{node}, P_{relay}=1, d=1000$, $l$ ranges in $(d/5,t_1=t_2,d/2,4d/5)$.}
    	\label{sim:fig6b}
    \end{subfigure}%
    \begin{subfigure}{0.33\textwidth}
    	\centering
    	\includegraphics[width=\linewidth]{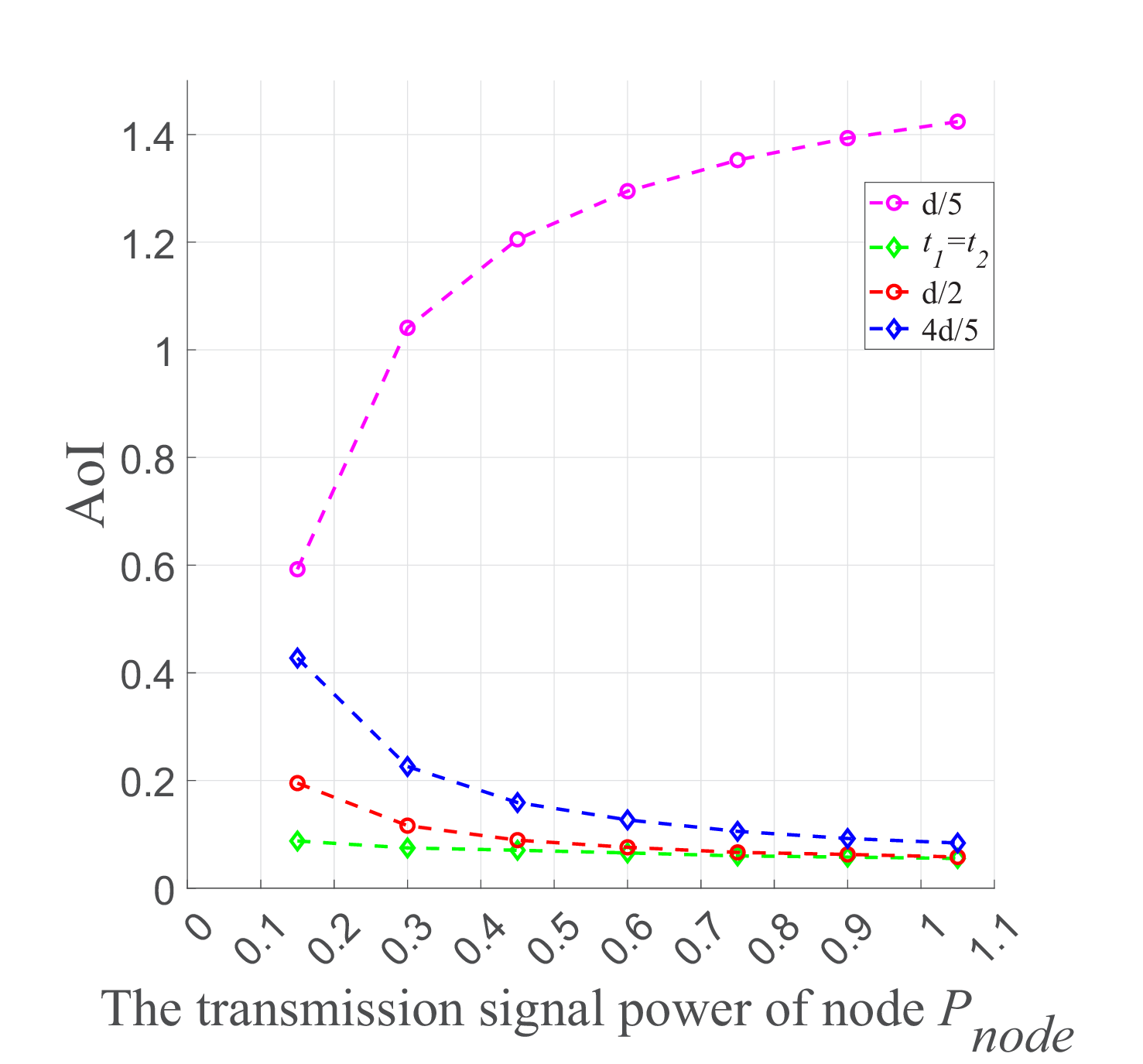}
    	\caption{Respect to different $P_{node}, P_{relay}=1.5, d=1500$, $l$ ranges in $(d/5,t_1=t_2,d/2,4d/5)$.}
    	\label{sim:fig6c}
    \end{subfigure}%
    \caption{Average Instant AoI respect to different $P_{node}, P_{relay}, d$, $l$ ranges in $(d/5,t_1=t_2,d/2,4d/5)$.}
    \label{sim:aoi2}
\vspace{-10pt}
\end{figure*}
%-------------------------------------------------------
\subsection{Results in Evaluation}
Fig.~\ref{sim:aoi1} presents a high degree of consistency between the experimental results and the mathematical reasoning, whether it is network delay or instant AoI. This undoubtedly positively demonstrates the reliability of our mathematical reasoning. Obviously, locating the relay at $t_1=t_2$ optimizes both the average network delay and the instant AoI.
Note that as the relay node is positioned closer to either the source or destination node, performance in terms of network delay and AoI significantly degrades. This is due to the weakening of the transmission signal over longer distances, which reduces data transmission rates. 

If the relay node is too close to the source node, although transmission efficiency in the first phase ({\it from the source to the relay}) improves, the efficiency in the second phase ({\it from the relay to the destination}) significantly decreases, causing the data packets transmitted in the first stage to be sent out faster and the data packets in the second stage to take longer to reach the destination, ultimately leading to worse network latency and AoI.

Similarly, if the relay node is positioned too close to the destination node, the transmission efficiency in the first phase ({\it from the source to the relay}) will be poor. Although the second-phase transmission ({\it from the relay to the destination}) becomes more efficient, it cannot compensate for the inefficiencies in the first phase, leading to deteriorated performance in terms of both the network delay and AoI.

Fig.~\ref{sim:aoi2} shows the performance of the optimal location on average instantaneous AoI
with respect to different $P_{node}$, $P_{relay}$, $d$, and $l$ ranges in ($d/5, t1 = t2, d/2, 4d/5$).
It is evident that the optimal location we extracted is indeed the best performing. 

It should be noted that when the relay is positioned at $d/5$ and the transmission signal power of the source node gradually increases, the AoI exhibits a rising trend. This is because the relay at this location is too close to the source node, and as the transmission signal power of the source node increases, data packets are transmitted more quickly. 
However, the longer transmission distance from the relay to the destination node causes the data packets output by the source node to wait. Meanwhile, the relay's transmission capacity does not substantially improve, resulting in increased waiting times for packets at the relay, which in turn leads to less fresh information. When the relay is located at other positions, this phenomenon of rapid packet transmission followed by long waiting times is less pronounced, and thus the AoI does not increase as the transmission signal power of the source node rises. This pattern can also be observed in Fig.~\ref{sim:aoi3} and Fig.~\ref{sim:aoi4}.

At the same time, we observed that the AoI at the relay position $d/2$ and the point where $t_1=t_2$ gradually converge as the transmission signal power of the source node increases. This is because, when the conditions of the first and second stages are similar, the location where $t_1=t_2$ aligns with the geographical midpoint.
%-----------------------------------------------------

As mentioned above, we assume that the noise level is uniform across each link, where the signal intensity of the white noise and the noise power are both set to $1\times 10^{-10}$. Next, we will check the cases where the noise levels
differ between links.

In the case where $N_1 > N_2$, the location that satisfies $t_1=t_2$ will be closer to the source node. As observed in Fig.~\ref{sim:aoi3},
even with added variables, the optimal location performs the best in terms of average instantaneous AoI. It is worth noting that, in this case, when the relay is positioned near the source node at $d/5$ and the transmission signal power of the source node is relatively low, the AoI performance is close to that of the optimal position and better than at other locations. This also indirectly reflects that the AoI performance is mathematically a continuous function.
\begin{figure*}
\centering
\begin{subfigure}{0.4\textwidth}
	\centering
	\includegraphics[width=\linewidth]{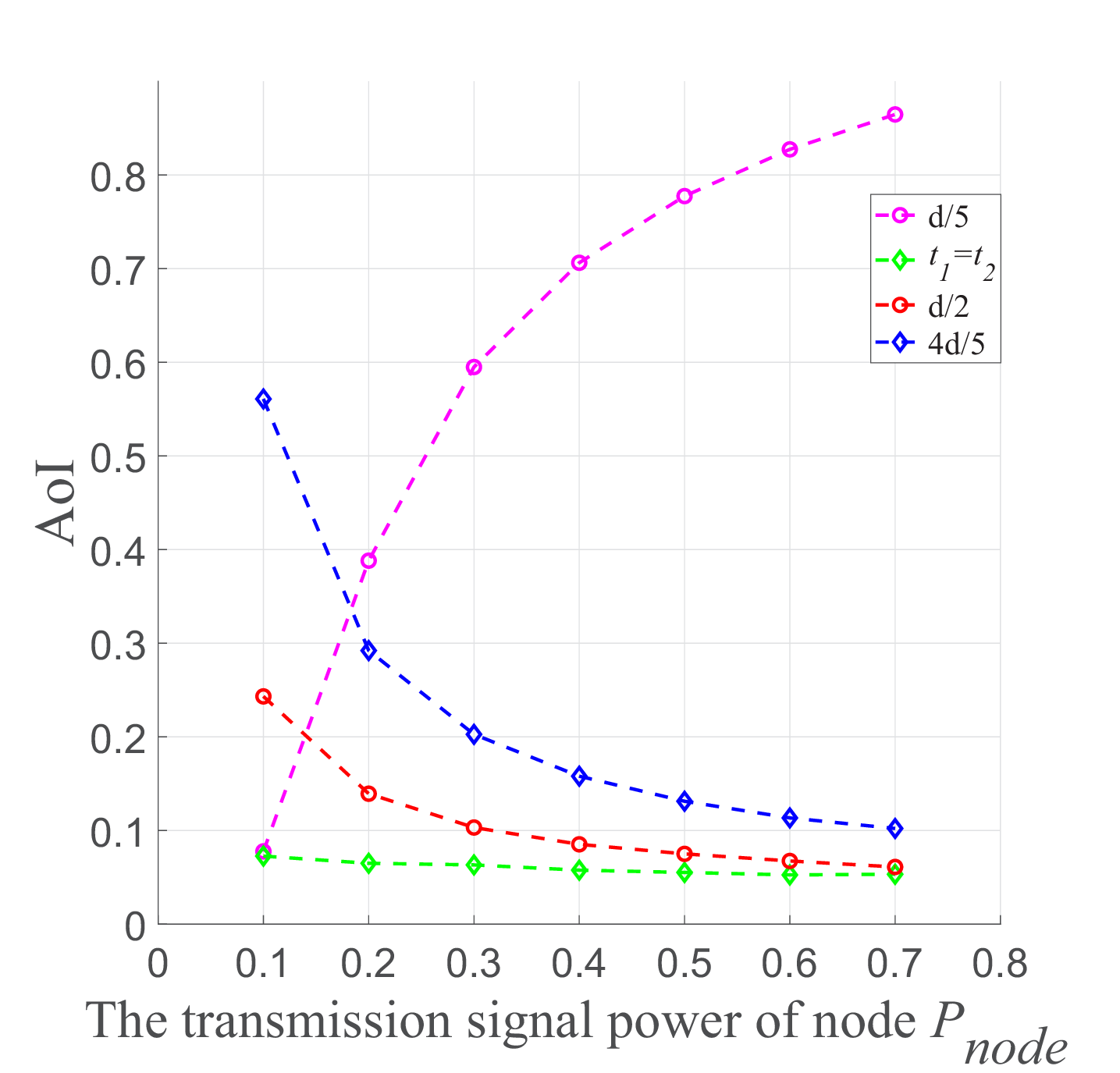}
	\caption{Respect to different $P_{node}, P_{relay}=1, d=1000$, $l$ ranges in $(d/5,t_1=t_2,d/2,4d/5)$.}
	\label{sim:fig7a}
\end{subfigure}%
\begin{subfigure}{0.4\textwidth}
	\centering
	\includegraphics[width=\linewidth]{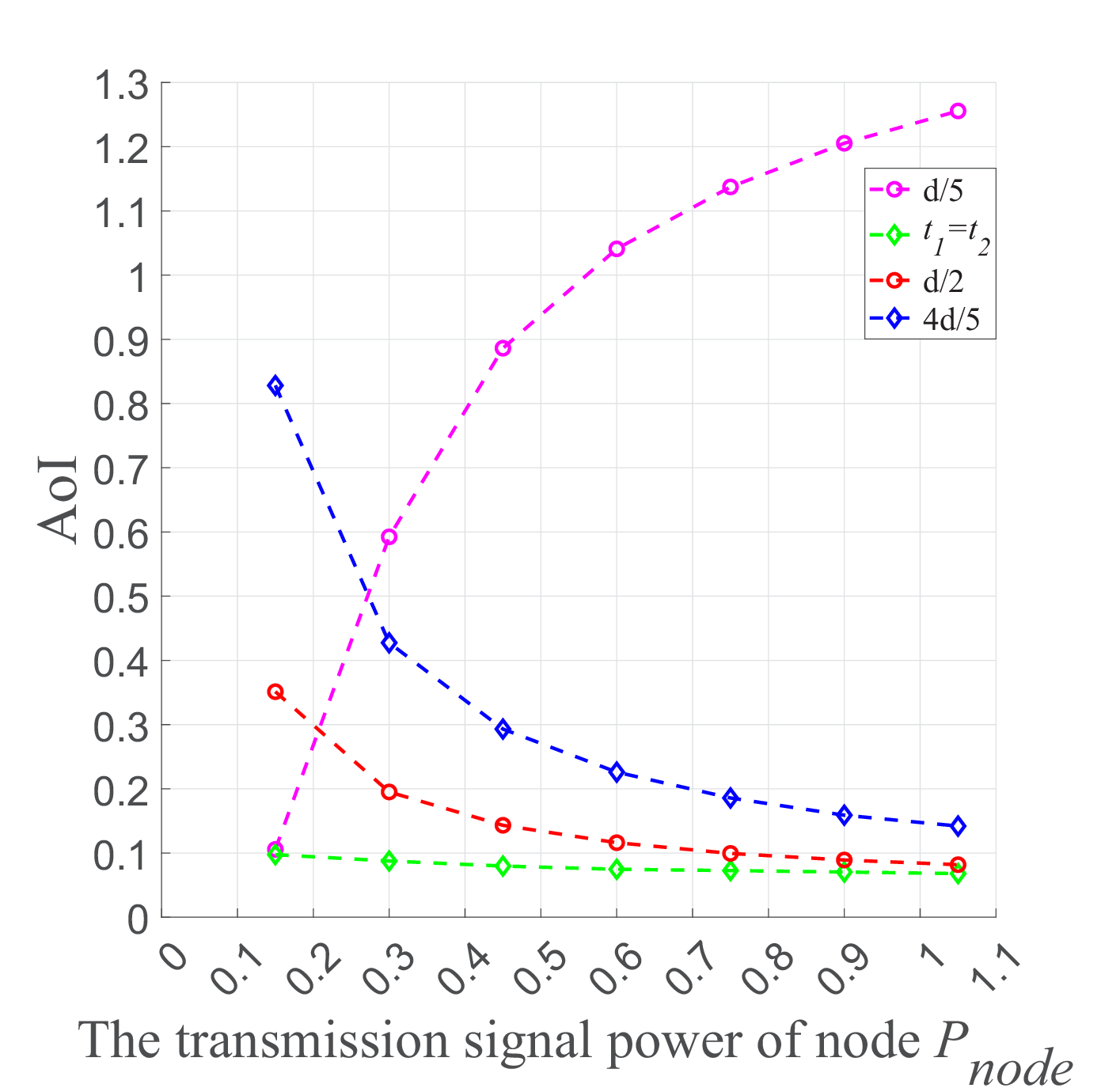}
	\caption{Respect to different $P_{node}, P_{relay}=1.5, d=1500$, $l$ ranges in $(d/5,t_1=t_2,d/2,4d/5)$.}
	\label{sim:fig7b}
\end{subfigure}%
\caption{Average Instant AoI respect to different $P_{node}, P_{relay}, d$, $l$ ranges in $(d/5,t_1=t_2,d/2,4d/5)$ with $N_1>N_2$.}
\label{sim:aoi3}
\vspace{-10pt}
\end{figure*}

%\begin{figure*}[ht!]
\begin{figure*}
\centering
\begin{subfigure}{0.4\textwidth}
	\centering
	\includegraphics[width=\linewidth]{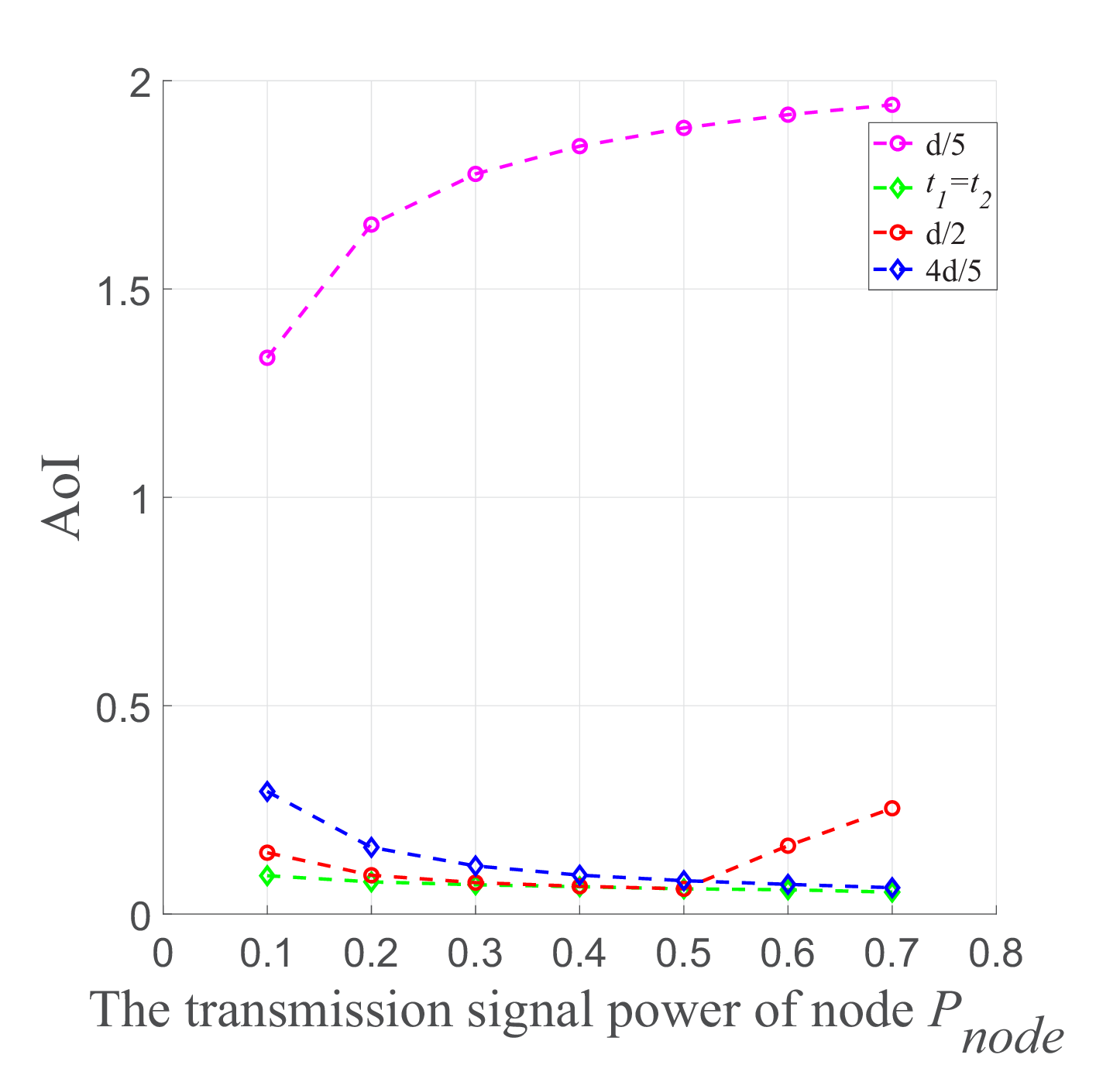}
	\caption{Respect to different $P_{node}, P_{relay}=1, d=1000$, $l$ ranges in $(d/5,t_1=t_2,d/2,4d/5)$.}
	\label{sim:fig8a}
\end{subfigure}%
\begin{subfigure}{0.4\textwidth}
	\centering
	\includegraphics[width=\linewidth]{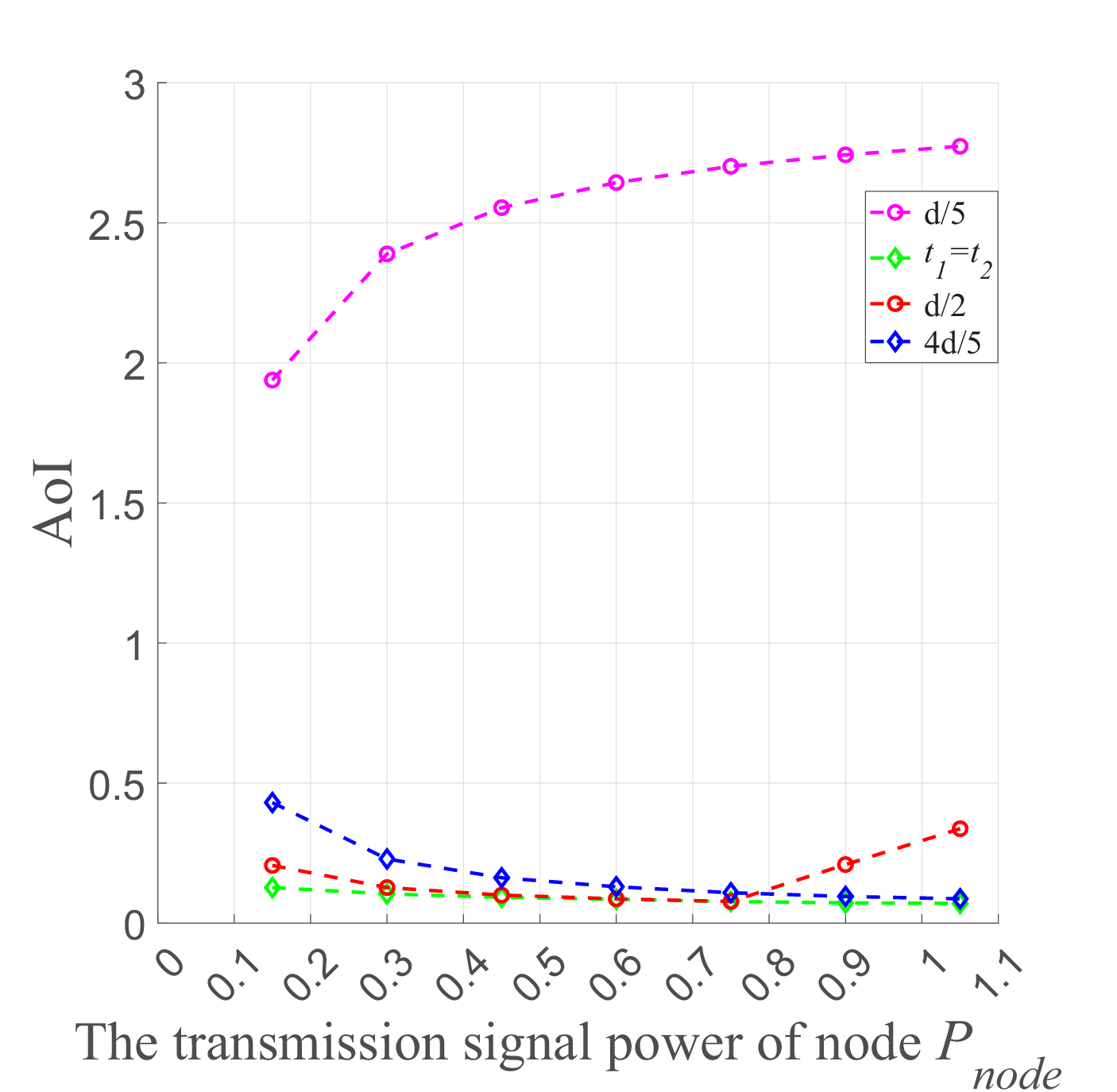}
	\caption{Respect to different $P_{node}, P_{relay}=1.5, d=1500$, $l$ ranges in $(d/5,t_1=t_2,d/2,4d/5)$.}
	\label{sim:fig8b}
\end{subfigure}%
\caption{Average Instant AoI respect to different $P_{node}, P_{relay}, d$, $l$ ranges in $(d/5,t_1=t_2,d/2,4d/5)$ with $N_1<N_2$.}
\label{sim:aoi4}
\vspace{-10pt}
\end{figure*}

In the case where $N_1 < N_2$, the location that satisfies $t_1=t_2$ will be closer to the destination node as $\frac{P_{node}}{P_{relay}}$ approaches ``\emph{1}''. Because $N_1 < N_2$, we need to shorten the distance between the relay and the destination node to balance the parameter's effect and determine the optimal location. It is shown in Fig.~\ref{sim:aoi4} that the AoI at the $d/2$ location gradually rises because, as the relay moves, the location at $d/2$ satisfies $t_1<t_2$, and its AoI increases instead in the case of $t_1<t_2$ as shown in (\ref{aoi_2}).

All cases demonstrate that the derived relay location consistently guarantees the optimal AoI and outperforms other schemes, as expected.

\section{Conclusions}\label{sec:conclusion}
In this paper, we focus on the AoI optimization problem in a typical post-disaster multi-hop wireless communication scenario,
where the data is transmitted via relays.
%in multi-hop communication with rescue equipment serving as relay nodes.
Based on the theoretical analysis,
we derive the analytical form of AoI,
and then derive the optimal AoI and its relay location.
%Considering the limited transmission capabilities of rescue equipment,
%we analyze how transmission distance affects signal attenuation.
Finally, through the extensive experiments,
our analysis conforms to the experimental results
and the optimal AoI is achieved.

Our findings, supported by mathematical reasoning and extensive simulations,
can reveal the importance of optimal relay placement in enhancing network performance (especially AoI).
%Further research will be conducted in the future, 
%focusing on enhancing network performance in scenarios involving simultaneous transmissions by multiple IoT devices.
The future research will focus on the multi-hop AoI optimization in network.

\section*{APPENDIX A  \\
}
%PROOF OF THEOREM 1
\textbf{Step 1)}: In the case of $t_1>t_2$, we use the derivative to analyze the trend of (\ref{aoi_3}). 
Let $F(l)=$(\ref{aoi_3}), where $l$ denotes the location of the relay node (the distance between the source node and the relay). We have
\begin{align}
F(l)=\frac{3}{2}\frac{L}{B\log_{2}(1+(\frac{\lambda}{4\pi l})^2\frac{G_t\ol{G_r}P_{node}}{N})}\nonumber\\
+\frac{L}{B\log_{2}(1+(\frac{\lambda}{4\pi (d-l)})^2\frac{\ol{G_t}G_rP_{relay}}{N})},
\label{F(l)}
\end{align}

\begin{align}
F'(l)=\frac{3\frac{\omega P_{node}}{1+\psi \frac{P_{node}}{l^2}}}{[\log_{2}(1+\psi \frac{P_{node}}{l^2})]^2l^3} \nonumber\\
-\frac{2\frac{\omega P_{relay}}{1+\psi \frac{P_{relay}}{(d-l))^2}}}{[\log_{2}(1+\psi \frac{P_{relay}}{(d-l)^2})]^2(d-l)^3},
\end{align}

where $\omega=\frac{L\psi}{B\ln 2}$, $\psi=\frac{\lambda^2G_t\ol{G_r}}{(4\pi)^2N}=\frac{\lambda^2\ol{G_t}G_r}{(4\pi)^2N}$. Let $F'(l)>0$, we have
\begin{align}
\frac{3}{2}\frac{P_{node}(d-l)^3}{P_{relay}l^3} \geq \frac{(1+\psi \frac{P_{node}}{l^2})(\log_2(1+\psi \frac{P_{node}}{l^2}))^2}{(1+\psi \frac{P_{relay}}{(d-l))^2})(\log_2(1+\psi \frac{P_{relay}}{(d-l))^2}))^2}.
\label{eqn:noeq}
\end{align}
First, we need to know from (\ref{eqn:t1}) (\ref{eqn:t2}), and (\ref{eqn:d}) that when $t_2\uparrow$, $t_1\downarrow$, $l\downarrow$ (or $t_2\downarrow$, $t_1\uparrow$, $l\uparrow$) is always true. When $t_1=t_2$, the right side of (\ref{eqn:noeq}) is always equal to $1$. Once $t_1>t_2$, as $t_1\uparrow$ and $l\uparrow$, the right side of (\ref{eqn:noeq}) is always less than $1$. Thus
\begin{align}
log_2\frac{3}{2}+log_2\frac{P_{node}}{P_{relay}}+log_2\frac{(d-l)^3}{l^3} \geq \log_21 =0.
\label{eqn:noeq2}
\end{align}
Solve this inequality, we have
\begin{align}
l\leq \frac{d}{1+(\frac{2}{3}\frac{P_{relay}}{P_{node}})^\frac{1}{3}}.
\label{eqn:noeq3}
\end{align}
Then we can conclude that when $l$ satisfies the above formula, $F(l)$ is an increasing function, as $l$ increases, $F(l)$ also increases.

Next, we will prove that $\frac{d}{1+(\frac{2}{3}\frac{P_{relay}}{P_{node}})^\frac{1}{3}}$ is always greater than $l_{t_1=t_2}$ (the location when $t_1=t_2$). That is, on the coordinate axis, (\ref{eqn:noeq3}) is always in the increasing direction of $l_{t_1=t_2}$.

When $t_1=t_2$, it is true that:

\begin{align} \label{eqn:noeq4}
\frac{P_{node}}{l_{t_1=t_2}^2}=\frac{P_{relay}}{(d-l_{t_1=t_2})^2}.
\end{align}

For $P_{relay}>P_{node}$, $\beta=P_{relay}/P_{node}>1$, $d-l_{t_1=t_2}=l_{t_1=t_2}\sqrt{\beta}$.

It implies that
$d-l_{t_1=t_2}+l_{t_1=t_2}=l_{t_1=t_2}\sqrt{\beta}+l_{t_1=t_2}=l_{t_1=t_2}(1+\sqrt{\beta})=d$, we have
\begin{align}
l_{t_1=t_2}=\frac{d}{1+\sqrt{\beta}}.
\end{align}
$\sqrt{\beta}=\beta^{1/2}$, $(\frac{2}{3}\beta)^{1/3}=(\frac{2}{3})^{1/3}\beta^{1/3}$, when $\beta>1$, $\beta^{1/2}>\beta^{1/3}$, and $(\frac{2}{3})^{1/3}<1$ then $\beta^{1/2}>(\frac{2}{3})^{1/3}\beta^{1/3}$, we have $1+\beta^{1/2}>1+(\frac{2}{3}\beta)^{1/3}$, that is 
\begin{align}
\frac{d}{1+\sqrt{\beta}}<\frac{d}{1+(\frac{2}{3}\beta)^{1/3}}
\end{align}
And it is always true that $l_{t_1=t_2}<\frac{d}{1+(\frac{2}{3}\frac{P_{relay}}{P_{node}})^\frac{1}{3}}$, which means that when $l_{t_1=t_2}<l<\frac{d}{1+(\frac{2}{3}\frac{P_{relay}}{P_{node}})^\frac{1}{3}}$ ($t_1>t_2$), $F(l)$ is increasing, $l_{t_1=t_2}$ is a minimum point of $\Delta(T)$.

\textbf{Step 2)}:
Here we analyze the changing trends of the following two functions:
\begin{equation}
		\begin{aligned}
			t_1 = \frac{L}{B\log_2(1+(\frac{\lambda}{4\pi})^2\frac{G_t\ol{G_r}}{N}\frac{P_{node}}{l^2})}\nonumber,
		\end{aligned}
\end{equation}
\begin{equation}
		\begin{aligned}
			t_2 = \frac{L}{B\log_2(1+(\frac{\lambda}{4\pi})^2\frac{\ol{G_t}G_r}{N}\frac{P_{relay}}{(d-l))^2})}\nonumber.
		\end{aligned}
\end{equation}

We first analyze the changing trends of $f_1(x)=\frac{1}{x^2}$:
$f_1'(x)=- \frac{2}{x^3}$, $f_1^{2}(x)=\frac{6}{x^4}$. It can be seen that when $x\uparrow$, $f_1'(x)\uparrow$, $f_1(x)$ decelerating $\downarrow$. Conversely, when $x\downarrow$, $f_1'(x)\downarrow$, $f_1(x)$ accelerating $\uparrow$. Let $f_2(x)=\frac{1}{x}$: $f_2'(x)=- \frac{1}{x^2}$, $f_2^{2}(x)=\frac{2}{x^3}$. It can be seen that when $x\uparrow$, $f_2'(x)\uparrow$, $f_2(x)$ decelerating $\downarrow$. Conversely, when $x\downarrow$, $f_2'(x)\downarrow$ and $f_2(x)$ accelerating $\uparrow$. $f_3(x)=\log_2(1+x)$: $f_3'(x)=\frac{1}{\ln{2}}\frac{1}{(1+x)}$, $f_3^{2}(x)=\frac{-1}{\ln{2}}\frac{1}{(1+x)^2}$. It can be seen that when $x\uparrow$, $f_3'(x)\downarrow$ and $f_3(x)$ decelerating $\uparrow$. Conversely, when $x\downarrow$, $f_3'(x)\uparrow$ and $f_3(x)$ accelerating $\downarrow$.

It is clear that $t_1$ increases as $l$ increases and $t_2$ decreases as $l$ increases. We find that no matter if $t_1$ is accelerating in increase, constant in rate of increase, or decelerating in increase, and $t_2$ is accelerating in decrease, constant in rate of decrease, or decelerating in decrease, there are at most two cases of $F(l)$. One case is that $F(l)$ will gradually increase all the time, the other is that $F(l)$  first increases and then decreases. By combining \textbf{Step 1}, we obtained the total change trend of $F(l)$: (1) $\nearrow$ of all the time; (2) $\nearrow$ followed by $\searrow$.  

Secondly, we compare $F(l)$ when $t_1=t_2$ and $l \rightarrow d$. Due to (\ref{F(l)}), when $t_1=t_2$, we have 
\begin{align} \label{eqn:t1=t2:1}
F(l_{t_1=t_2})=\frac{5}{2}\frac{L}{B\log_{2}(1+(\frac{\lambda}{4\pi l_{t_1=t_2}})^2\frac{G_t\ol{G_r}P_{node}}{N})}.
\end{align}
Draw from (\ref{eqn:noeq4}), as $P_{relay}>P_{node}$, $l_{t_1=t_2}<\frac{d}{2}$, then
\begin{align} \label{eqn:t1=t2:2}
F(l_{t_1=t_2})<\frac{5}{2}\frac{L}{B\log_{2}(1+4(\frac{\lambda}{4\pi d})^2\frac{G_t\ol{G_r}P_{node}}{N})}.
\end{align}
When $l \rightarrow d$, we have  
\begin{align} \label{eqn:l=d}
F(l)_{l\rightarrow d}\triangleq \frac{3}{2}\frac{L}{B\log_{2}(1+(\frac{\lambda}{4\pi d})^2\frac{G_t\ol{G_r}P_{node}}{N})}.
\end{align}
Let $F(l_{t_1=t_2}) < F(l)_{l\rightarrow d}$, we get the following inequality:
\begin{align} \label{eqn:leq}
(1+4x)^{\frac{3}{5}}>1+x,
\end{align}
where $x=\frac{\lambda^2G_t\ol{G_r}P_{node}}{(4\pi d)^2N}$, and we know that $x$ is always $>0$. When $x>0$ and is small enough, (\ref{eqn:leq}) is true. When most of the parameters that make up $x$ are fixed, it depends on the changes in $P_{node}$, $d$, and $N$. After calculation, we found that $x$ is generally small enough under the control of $P_{node}$, $d$, and $N$, which are dynamic data. The lower bound of $F(l)_{l\rightarrow d}$ leads to $t_1=t_2\in\argmin_{t_1,t_2} \Delta(T)$. The proof is complete.

\iffalse
It can be observed that when $t_1\uparrow$, $d_1\uparrow$, $(\frac{\lambda}{4\pi})^2\frac{G_t\ol{G_r}}{N}\frac{P_{node}}{d_1^2}\downarrow$ decelerated, $\log_2(1+(\frac{\lambda}{4\pi})^2\frac{G_t\ol{G_r}}{N}\frac{P_{node}}{d_1^2})\downarrow$ accelerated, $t_1\uparrow$ accelerated. When $t_1\uparrow$, $d_2\downarrow$, $(\frac{\lambda}{4\pi})^2\frac{\ol{G_t}G_r}{N}\frac{P_{relay}}{d_2^2})\uparrow$ accelerated, $\log_2(1+(\frac{\lambda}{4\pi})^2\frac{\ol{G_t}G_r}{N}\frac{P_{relay}}{d_2^2})\uparrow$ decelerated, $t_2\downarrow$ decelerated. From \textbf{Step 1}, we know that function $F(l)$ is monotonically increasing in the lower range. At the same time, $t_1$ shows accelerated growth and $t_2$ shows deceleration reduction. That is to say, if this trend is maintained, $F(l)$ will continue to increase. The proof is complete.
\fi

\ifCLASSOPTIONcaptionsoff
\newpage
\fi

\bibliographystyle{IEEEtran}
\bibliography{myReferences}

\end{document}

%% file: Supporting_Preambles/symbols_commands.tex
\newcommand{\wt}{\widetilde{\mathbf{w}}}

\newcommand{\x}{\mathbf{x}}

\newtheorem{thm}{Theorem}